\documentclass[a4paper,twoside,11pt,reqno]{amsart}
\usepackage{fullpage}
\usepackage[T1]{fontenc}
\usepackage[utf8]{inputenc}

\usepackage{amsmath,amsfonts,amsbsy,amsgen,amscd,mathrsfs,amssymb,amsthm,bm}
\usepackage{enumerate,mathtools,comment}

\usepackage[usenames,dvipsnames]{xcolor}
\usepackage[colorlinks=true,citecolor=blue,linkcolor=blue]{hyperref}
\usepackage{braket}
\usepackage{tikz}
\usetikzlibrary{shadings}

\usepackage{pagecolor}

\newtheorem{theorem}{Theorem}
\newtheorem{lemma}[theorem]{Lemma}

\newtheorem{defn}{Definition}
\newtheorem*{defn*}{Definition}
\newtheorem*{prop*}{Proposition}

\newtheorem*{conj*}{Conjecture}

\newtheorem*{fact*}{Fact}
\newtheorem{prop}{Proposition}
\newtheorem*{ex*}{Example}

\newtheorem{rem}{Remark}

\newcommand{\lS}{{\bs\alpha}}
\newcommand{\ls}{\alpha}
\newcommand{\lT}{{\bs\beta}}
\newcommand{\lt}{\beta}
\newcommand{\lZ}{{\bs\omega}}
\newcommand{\lz}{\omega}

\newcommand{\beps}{{\bs\epsilon}}

\newcommand*{\claimproofname}{Proof of claim}

\DeclareMathOperator{\poly}{poly}

\DeclareMathOperator{\tr}{tr}

\newcommand{\al}{\alpha}
\newcommand{\eps}{\varepsilon}

\newcommand{\cT}{{\mathcal T}}

\newcommand{\mc}[1]{\mathcal{#1}}
\newcommand{\bs}[1]{\boldsymbol{#1}}

\newcommand{\bC}{\mathbf{C}}

\newcommand{\bH}{\mathbf{H}}

\newcommand{\cH}{\mathcal{H}}

\newcommand{\ketbra}[2]{\ket{#1}\!\!\bra{#2}}
\DeclareMathSymbol{\shortminus}{\mathbin}{AMSa}{"39}

\renewcommand{\epsilon}{\varepsilon}

\theoremstyle{definition}

\numberwithin{equation}{section} 
\numberwithin{figure}{section}
\numberwithin{table}{section}

\newcommand{\bE}{\mathbf{E}}
\newcommand{\R}{\mathbf{R}}
\newcommand{\C}{\mathbf{C}}
\newcommand{\bP}{\mathbf{P}}

\newcommand{\D}{\mathbf{D}}

\newcommand{\Z}{\mathbf{Z}}

\newcommand{\un}{\mathbf{1}}

\newcommand{\x}{\bs{x}}

\newcommand{\z}{\bs{z}}

\newcommand{\X}{\mathbf{X}}
\newcommand{\Y}{\mathbf{Y}}
\renewcommand{\a}{\bs{a}}

\newcommand{\Rad}{\textnormal{Rad}}
\newcommand{\cP}{\mathcal{P}}

\newcommand{\cS}{\mathcal{S}}
\newcommand{\ce}{\mathcal{E}}
\newcommand{\cA}{\mathcal{A}}

\newcommand{\bh}{\textnormal{BH}}

\newcommand{\go}{K}
\newcommand{\px}{a}
\newcommand{\pz}{b}
\newcommand{\vpx}{\textbf{a}}
\newcommand{\vpz}{\mathbf{b}}
\newcommand{\gx}{s}
\newcommand{\gz}{t}
\newcommand{\vgx}{\mathbf{s}}
\newcommand{\vgz}{\mathbf{t}}
\newcommand{\ca}{\mathcal{A}}
\newcommand{\spa}{\textnormal{span}}

\newcommand{\1}{\mathbf{1}}

\begin{document}
	
	\title[Quantum norm designs]{Approximating the operator norm of local Hamiltonians via few quantum states}

\author[L. Becker]{Lars Becker}
\address{(L.B.) Mathematics Department, Princeton University, Fine Hall, Washington Road, Princeton, NJ 08544-1000, USA}
\email{lbecker@math.princeton.edu}

\author[J. Slote]{Joseph Slote}
\address{(J.S.) Department of Computing and Mathematical Sciences, California Institute of Technology, Pasadena, CA}
\email{jslote@caltech.edu}

\author[A. Volberg]{Alexander Volberg}
\address{(A.V.) Department of Mathematics, MSU, 
East Lansing, MI 48823, USA, and Hausdorff Center of Mathematics, Bonn}
\email{volberg@math.msu.edu}

\author[H. Zhang]{Haonan Zhang}
\address{(H.Z.) Department of Mathematics, University of South Carolina}
\email{haonanzhangmath@gmail.com}

\begin{abstract}
	Consider a Hermitian operator $A$ acting on a complex Hilbert space of dimension $2^n$. We show that when $A$ has small degree in the Pauli expansion, or in other words, $A$ is a local $n$-qubit Hamiltonian, its operator norm can be approximated independently of $n$ by maximizing $|\braket{\psi|A|\psi}|$ over a small collection $\mathbf{X}_n$ of product states $\ket{\psi}\in (\mathbf{C}^{2})^{\otimes n}$. More precisely, we show that whenever $A$ is $d$-local, \textit{i.e.,} $\deg(A)\le d$, we have the following discretization-type inequality:
	\[
\|A\|\le C(d)\max_{\psi\in \mathbf{X}_n}|\braket{\psi|A|\psi}|.
	\]
The constant $C(d)$	depends only on $d$. This collection $\mathbf{X}_n$ of $\psi$'s, termed a \emph{quantum norm design}, is independent of $A$, and
can have cardinality as small as $C^n$, which is essentially tight. Previously, norm designs were known only for homogeneous $d$-local $A$ \cite{L,BGKT,ACKK}, and for non-homogeneous $2$-local traceless $A$ \cite{BGKT}.
	Several other results, such as boundedness of Rademacher projections for all levels and estimates of operator norms of random Hamiltonians, are also given.

	\end{abstract}
	

\thanks{
		}

	\maketitle
	

	\section{Introduction}
Let $\mc H=\C^2$ denote a two-dimensional complex Hilbert space and consider $A$ a Hermitian operator (or \textit{Hamiltonian}) on $\mc H^{\otimes n}$. In many problems in quantum physics and quantum computer science \cite{KKR,BGKT,KSV}, it is important to approximate the operator norm of $A$ 
\begin{equation}\label{eq:duality}
	\|A\|
	:=\sup_{\ket{\psi}}|\braket{\psi|A|\psi} |
		=\sup_{\rho}|\tr[A\rho]|
	\end{equation}
	where $\ket{\psi}$ is any unit vector in $\cH^{\otimes n}$ and $\rho$ is any density operator on $\cH^{\otimes n}$.
	
Computing $\|A\|$ is a hard problem in general when $n$ is large. In this work, we will focus on the case when $A$ is \textit{local}. Recall that any operator $A$ on $\mc H^{\otimes n}$ has the unique Pauli expansion
	\begin{equation}\label{eq:fourier expansion}
	A=\sum_{\lS\in \{0,1,2,3\}^n}\widehat{A}_{\lS}\sigma_{\ls_1}\otimes \cdots \otimes \sigma_{\ls_n}
 =\sum_{\lS\in \{0,1,2,3\}^n}\widehat{A}_{\lS}\sigma_{\lS},
	\end{equation}
	where $\sigma_0=\un$ is the $2$-by-$2$ identity matrix, and $\sigma_j,j=1,2,3$ are the Pauli matrices 
	\begin{equation*}
	 \sigma_1=\begin{pmatrix}0&1\\1&0\end{pmatrix},\quad \sigma_2=\begin{pmatrix}0&-i\\i&0\end{pmatrix},\quad \sigma_3=\begin{pmatrix}1&0\\0&-1\end{pmatrix}
	\end{equation*}
that satisfy anti-commutation relation
\begin{equation}\label{eq:anti}
\sigma_i\sigma_j+\sigma_j\sigma_i=2\delta_{ij}\un,\qquad 1\le i,j\le 3.
\end{equation}
In \eqref{eq:fourier expansion}, $\widehat{A}_{\lS}\in\C$ is the Pauli coefficient, and $\sigma_{\lS}=\sigma_{\ls_1}\otimes \cdots \otimes \sigma_{\ls_n}$ for $\lS=(\ls_1,\dots, \ls_n)\in \{0,1,2,3\}^n$ denote the \emph{Pauli monomials}. For a positive integer $d\ge 1$, we define the degree of $A$ as
$$
\deg(A):=\max\{|\lS|:\widehat{A}_{\lS}\neq 0\}
$$
where $|\lS|:=|\{j: \ls_j\neq 0\}|$. Here and in what follows, for a set $S$, $|S|$ denotes its cardinality.
By saying $A$ is local, we mean $\deg(A)$ is small:
in general, for a positive integer $d$, we say $A$ is \textit{$d$-local} if $\deg(A)\le d$.
We will also say $A$ is $d$-\emph{homogeneous} if $\widehat{A}_{\lS}$ is nonzero only when $|\lS|=d$ (rather than just $|\lS| \leq d$).

One may hope for better methods to compute $\|A\|$ when $A$ is local, but this is still challenging.
In fact, given a 2-local Hamiltonian $A$ with $\|A\|\leq 1$, deciding whether $\|A\|$ is at most $a\in (0,1)$ or at least $b\in (a,1)$ for $|b-a|\leq 1/\poly(n)$ is $\mathbf{QMA}$-Complete.
	$\mathbf{QMA}$ (Quantum Merlin Arthur) is a computational complexity class that is the natural quantum analogue of $\mathbf{NP}$ \cite{GHLS, W}.
	
In this work we show that $\|A\|$ can be approximated up to a multiplicative constant free of $n$, provided that $A$ is local, by considering \eqref{eq:duality} over a ``small'' set of states $\ket{\psi}$ or $\rho$.
We call sets of states allowing for these comparisons \textit{quantum norm designs}.

\begin{defn}
    Let $\X=\X_1,\X_2,\ldots$ be a sequence of sets with $\X_n$ denoting a set of $n$-qubit quantum states.
    We call $\X$ a \textit{quantum norm design} if there exists a constant $C(d)$ depending on $d$ but not $n$ such that for all $n$ and all $n$-qubit degree-d operators $A$,
    \[\sup_{\rho\in \X_n}\lvert\tr[A\rho]\rvert\quad\leq\quad\|A\|\quad\leq\quad C(d)\sup_{\rho\in \X_n}\lvert\tr[A\rho]\rvert\,.\]
\end{defn}
Note that the left-hand side is trivial. Recalling that each $\sigma_j,j=1,2,3$ has $\pm 1$ as eigenvalues, we use $D$ to denote the set of eigenstates of $\sigma_j,j=1,2,3$ corresponding to $\pm1$. Then $|D|=6$.

\begin{theorem}
    \label{thm: quantum norm design}
    Let $D$ be as above.
    Then for all degree-$d$ Hermitian operators $A$ on $\mathcal H^{\otimes n}$,
    \begin{equation}
    \label{ineq:basic-design-ineq}
    \|A\|\leq \tfrac32(3+3\sqrt{2})^d\sup_{\psi\in D^{\otimes n}}\lvert\tr[A\psi]\rvert\,.
    \end{equation}
    That is, $D^{\otimes n}=\{\otimes_{j=1}^n\psi_j\}_{\psi_j\in D, 1\le j\le n}$ for $n=1,2,\ldots$ is a quantum norm design with constant $C(d)=\frac32(3+3\sqrt{2})^d$. Moreover, if $A$ is $d$-homogeneous, we can take a better constant $C(d)=3^d$.
\end{theorem}

For \textit{homogeneous} 2-local Hamiltonians, Lieb \cite{L} proved a result of the type \eqref{ineq:basic-design-ineq} with multiplicative constant $C_1= 9$.
This was extended to general $2$-local Hamiltonians by Bravyi--Gosset--K\"onig--Temme \cite{BGKT}, who obtained \eqref{ineq:basic-design-ineq} for $2$-local traceless Hamiltonians with the same constant $9$, using a nice idea that allows them to reduce the problem to the homogeneous case studied by Lieb.
We recall their idea  as it is also used in our proof of Figiel's inequality for qubit systems that is discussed in Sect. \ref{sec:Figiel}.

Also implicit in the work of \cite{BGKT} is a proof that tensor products of Pauli eigenstates are quantum norm designs for homogeneous $A$ of general degree with $C(d)=3^d$ (the reader may also consult the Appendix E of \cite{ACKK} where the proof is worked out by Anschuetz, Chen, Kiani and King in full).
Theorem \ref{thm: quantum norm design} extends this line of work to include non-homogeneous $d$-local Hamiltonians.

We next study what flexibility there is in choosing $\mathbf{X}_n$ satisfying the requirements of a quantum norm design, both in terms of the cardinality of $\mathbf{X}_n$ and the geometry of its constituent states.
In Sect. \ref{card} we show in Theorem \ref{thm: small quantum norm design} that in the limit of large $n$, the cardinality of quantum norm designs can be improved from $6^n$ to $C(\epsilon)(1+\epsilon)^n$ for any $\epsilon>0$ by subsampling our candidate norm design from Theorem \ref{thm: quantum norm design}.
We also show this cardinality is essentially optimal, even for norm designs not composed of product states---this is Theorem \ref{thm: lower bound design}.
In Sect. \ref{sec:2design} we further study the geometry of $\mathbf{X}_n$'s by showing that tensor powers of any \emph{1-qubit 2-design} also constitute norm designs.

\medskip

The norm design terminology is inspired by \textit{quantum state designs} \cite{AE}, or more originally the \textit{spherical designs} of Delsarte and Goethals \cite{DGS}, which refer to any discrete sets of points on the sphere, the uniform measure over which reproduces the uniform measure on the whole sphere for low-degree polynomials.
In comparison, here we are only concerned with ``reproducing the operator norm'' for low-degree operators in the asymptotic (and approximate) sense of a dimension-free estimate.


\medskip

Ideas developed in proving the results above also allow us to improve the constant in the \emph{Bohnenblust--Hille inequality} for qubit systems, the central technical result behind recent progress in learning bounded local Hamiltonians and observables \cite{HCP,VZ}.
See Sect. \ref{sec:multilinear} for details. The inequality of Figiel, as well as other related dimension-free inequalities, are discussed in Sect. \ref{sec:Figiel}.

\subsection*{Related work}
We conclude by mentioning some other related work.
Gharabian and Kempe \cite{GK} studied approximation ratio
with respect to the maximal eigenvalue of a local Hamiltonian
which is a sum of positive semidefinite terms.
 Brandao and Harrow \cite{BrH} established
upper bounds on the additive error between the
energy attainable by a product state and the maximal
eigenvalue.

Another relevant result by Harrow and Montanaro \cite{HM} gave an algorithm that given a
traceless $2$-local Hamiltonian $A$ of the form 
$$
A= \sum_{|\lS|=1,2} \widehat{A}_{\lS}\sigma_{\lS}
$$
outputs a product state $\ket \phi$ with energy 
$$
 \sum_{|\lS|=1,2} |\widehat{A}_{\lS}| \lesssim  n\bra \phi A\ket\phi \,.
$$
Let us remark that in the previous work on the aforementioned Bohnenblust--Hille inequality, Volberg and Zhang \cite{VZ} proved that in this setting
\begin{equation}
\label{quarter}
\sum_{|\lS|=1,2} |\widehat{A}_{\lS}| \lesssim n^{1/2}\|A\| \,.
\end{equation}
which comes out as a combination of \cite{VZ} and Fourier analysis results on discrete hypercubes going back to the celebrated Littlewood's $4/3$ inequality \cite{Lit}.


\medskip

Finally, we also mention that Theorem \ref{thm: quantum norm design} bears some similarity to a family of results in the classical approximation theory literature known as \textit{Bernstein-type discretization inequalities}, or \textit{discretizations of the uniform norm}.
Here one seeks to control the supremum norm of a low-degree multivariate polynomial $p$ over some domain $\Omega$ by its absolute supremum over some finite subset $X\subset \Omega$.
Relevant recent work is \cite{DP,BKSVZ}, the latter of which contains some estimates that we will find useful in the sequel.


\subsection*{Notation}
Dirac bra-ket notation will be used for quantum states.
For pure states $\ket{\psi}$ we will use $\psi$ to denote the rank-one projector onto $\ket{\psi}$, \textit{i.e.}, $\psi:=\ketbra{\psi}{\psi}$.

\subsection*{Acknowledgments}
Part of this work was finished when the authors were visiting the program \textit{Boolean Analysis in Computer Science}, at the Hausdorff Institute for Mathematics in the fall of 2024, funded by the Deutsche Forschungsgemeinschaft (DFG, German Research Foundation) under Germany’s Excellence
Strategy – EXC-2047/1 – 390685813. We are grateful to their kind hospitality. 
L.B. is supported by the Collaborative Research Center 1060 funded by the Deutsche Forschungsgemeinschaft and the Hausdorff Center for Mathematics, funded by the DFG under Germany's Excellence Strategy - GZ 2047/1. J.S. is supported by Chris Umans' Simons Institute Investigator Grant. A.V. is supported by NSF DMS-1900286, DMS-2154402 and by Hausdorff Center for Mathematics. H.Z. is supported by NSF DMS-2453408. The third author is grateful to Christopher Baldwin for indicating to him the paper \cite{AGK}. He is grateful to Mark Dykman and Christopher Baldwin for valuable discussions. 

\section{Proof of main results: quantum norm designs}
	\label{sec:multilinear}

In this section we prove the main result, Theorem \ref{thm: quantum norm design}. For any positive integer $n$, we put $[n]:=\{1,2,\dots,n\}.$ We start with some lemmas. 

\medskip

Recall that each Pauli matrix $\sigma_\ls$, $\ls=1,2,3$, has $\pm1$ as eigenvalues. For $\ls=1,2,3$ and $ \eps=\pm 1$, let $\ket{e^{(\ls)}_\eps}$ be the unit eigenvector of $\sigma_\ls$ with eigenvalue $\eps$. One useful property is the following. See also \cite[Lemma 2.1]{VZ}.

\begin{lemma}
\label{lem:anti}
For $\epsilon\in\{-1,1\}$, we have 
\begin{equation}\label{eq:eigenproj}
\ketbra{e^{(\ls)}_\epsilon}{e^{(\ls)}_\epsilon}
=\frac{1}{2}\sigma_0+\frac{1}{2}\epsilon\sigma_\ls,\qquad \ls
=1,2,3
\end{equation}
and 
\begin{equation}\label{eq:orth}
\tr[\sigma_{\ls} \ketbra{e^{(\lt)}_\epsilon}{e^{(\lt)}_\epsilon}]=\epsilon\delta_{\ls\lt},\qquad \ls,\lt=1,2,3.
\end{equation}
Moreover, for each $\ls\in\{1,2,3\}$, there exists a unitary $U_\ls$ over $\mc H=\C^2$ such that $U_\ls^\ast \sigma_\ls U_\ls=\sigma_3$, and thus 
\begin{equation}
U_\ls^\ast \ketbra{e^{(\ls)}_\epsilon}{e^{(\ls)}_\epsilon} U_\ls
=U_\ls^\ast \left(\frac{1}{2}\sigma_0+\frac{1}{2}\epsilon\sigma_\ls\right) U_\ls
=\frac{1}{2}\sigma_0+\frac{1}{2}\epsilon\sigma_3.
\end{equation}
\end{lemma}

\begin{proof}
The first equation \eqref{eq:eigenproj} is a direct computation. 
The second equation \eqref{eq:orth} follows from \eqref{eq:eigenproj} and the fact that $\sigma_{\ls},0\le \ls\le 3$ are orthonormal with respect to the inner product given by the normalized trace $\frac{1}{2}\tr$. The last statement follows from the fact that each of $\sigma_\ls,\ls=1,2,3$ is Hermitian with eigenvalues $\pm 1$.
\end{proof}

The main difficulty of proving Theorem \ref{thm: quantum norm design} is the non-commutativity of generic pairs of $\sigma_{\lS}$'s, $\lS\in\{0,1,2,3\}^n$.
If $A$ happens to be a linear combination of $\sigma_{\lS}$'s that commute, one may use their common eigenprojections as $\mathbf{X}_n$. To deal with the general case, we will organize $\sigma_{\lS}$'s according to a partial order on indices.

	\begin{defn}[Partial order on Pauli monomials]
	\label{defn:pauli-partial-order}
		Let $\lS,\lT\in\{0,1,2,3\}^n$.
		Then we say $\lS\leq \lT$ if for all $j=1,2,\ldots, n$ it holds that $\ls_j\in\{0,\lt_j\}$.
	\end{defn}
	
	Note that the maximal elements with respect to $\leq$ are all the $\lZ\in[3]^n=\{1,2,3\}^n$, corresponding to the Pauli monomials $\sigma_{\lZ}$ of maximum degree.

For all $\lZ\in[3]^n$ and $\beps\in\{-1, 1\}^n$, define 
\begin{equation}\label{r-s}
\rho_{\beps, \lZ}
:=\ketbra{e^{\lz_1}_{\eps_1}}{e^{\lz_1}_{\eps_1}}\otimes\dots\otimes  \ketbra {e^{\lz_n}_{\eps_n}}{e^{\lz_n}_{\eps_n}} 
\stackrel{\eqref{eq:eigenproj}}{=}\Big(\frac12\sigma_0+ \frac12 \eps_1 \sigma_{\lz_1}\Big)\otimes\dots\otimes  \Big(\frac12\sigma_0+ \frac12 \eps_n \sigma_{\lz_n}\Big),
\end{equation}
\textit{i.e.,} the tensor product of eigenprojections of $\sigma_{\lz_j},j\in [n]$ corresponding to eigenvalues $\eps_j$, $j\in [n]$. For any $\lZ\in[3]^n$, consider
the map
\begin{equation}\label{defn:es}
\ce_{\lZ}(A):=\sum_{\beps\in\{-1,1\}^n}\rho_{\beps,\lZ}A\rho_{\beps,\lZ}.
\end{equation}
The operator $\ce_{\lZ}$ is the conditional expectation onto the commutative subalgebra $\mathcal{A}_{\lZ}$ generated by 
$$\un\otimes \cdots \otimes \sigma_{\lz_j}\otimes  \cdots \otimes \un,\qquad j\in [n]$$
where $\sigma_{\lz_j}$ appears in the $j$-th place. 
It also is related to the $n$-fold tensor product of the 1-qubit depolarizing channel with parameter 1/3 employed by \cite{BGKT} via averaging over $\lZ$'s, as explained as part of the next lemma.

\begin{lemma}
For any $\lS\in\{0,1,2,3\}^n$ and $\lZ\in[3]^n$, $\ce_{\lZ}$ is a conditional expectation such that 
\begin{equation}\label{es:key1}
\ce_{\lZ}(\sigma_{\lS})=
\begin{cases}
\sigma_{\lS}, & \lS\le \lZ\\
0, & \lS\not\le \lZ
\end{cases}
\end{equation}
and 
\begin{equation}\label{es:key2}
\frac{1}{3^n}\sum_{\lZ\in[3]^n}\ce_\lZ(\sigma_{\lS})=3^{-|\lS|}\sigma_{\lS}.
\end{equation}
As a consequence, for any $A=\sum_{\lS}\widehat{A}_{\lS}\sigma_{\lS}$ one has
\begin{equation}\label{es:key1 general}
\ce_{\lZ}(A)=\sum_{\lS:\lS\le \lZ}\widehat{A}_{\lS}\sigma_{\lS}
\end{equation}
and
\begin{equation}\label{es:key2 general}
\frac{1}{3^n}\sum_{\lZ\in[3]^n}\ce_\lZ(A)=\sum_{\lS}3^{-|\lS|}\widehat{A}_{\lS}\sigma_{\lS}.
\end{equation}
Moreover, 
\begin{equation}
\label{AsA}
\tr [A \,\rho_{\beps, \lZ}] = \tr [\ce_\lZ(A)\, \rho_{\beps,\lZ}], \quad  \beps\in \{-1,1\}^n, \, \lZ\in [3]^n.
\end{equation}
\end{lemma}

\begin{rem}
\label{scenario}
One can call $\lZ$ a {\it scenario}, it is the same as a map $s: [n]\to [3]$, and $\ce_{\lZ}(A)$ gives us the sum of monomials of $A$ such that on $i$-th place
monomials have either $\sigma_0$ or $\sigma_{s(i)}$.
\end{rem}

\begin{proof}
By definition, $\ce_{\lZ}$ is linear and completely positive. By \eqref{es:key1}, it is unital and $\ce_{\lZ}^2=\ce_{\lZ}$, thus a conditional expectation. The equations \eqref{es:key1 general} and \eqref{es:key2 general} are immediate consequences of \eqref{es:key1} and \eqref{es:key2} by linearity. The identity \eqref{AsA} is a consequence of the fact that $\ce_{\lZ}$ is a conditional expectation, since $\rho_{\beps,\lZ}$ belongs to the commutative subalgebra $\mathcal{A}_{\lZ}$. Or, one can see \eqref{AsA} from linearity, \eqref{es:key1} and \eqref{prod} below. So, it suffices to verify \eqref{es:key1} and \eqref{es:key2}.
 
 To see \eqref{es:key1}, note that 
 $$
 \ce_{\lZ}(\sigma_{\lS})=
 \sum_{\eps}\prod_{j\in [n]}\bra{e^{\lz_j}_{\epsilon_j}}\sigma_{\ls_j}\ket{e^{\lz_j}_{\epsilon_j}}\cdot \ketbra{e^{\lz_1}_{\epsilon_1}}{e^{\lz_1}_{\epsilon_1}}\otimes \cdots\otimes \ketbra{e^{\lz_n}_{\epsilon_n}}{e^{\lz_n}_{\epsilon_n}}.
 $$
Recall that by Lemma \ref{lem:anti}
$$
\bra{e^{\lz_j}_{\epsilon_j}}\sigma_{\ls_j}\ket{e^{\lz_j}_{\epsilon_j}}
=
\begin{cases}
1, & \ls_j=0\\
\epsilon_j, & \ls_j=\lz_j\\
0, &\textnormal{otherwise}
\end{cases}
$$
which implies 
\begin{equation}\label{prod}
\prod_{j\in [n]}\bra{e^{\lz_j}_{\epsilon_j}}\sigma_{\ls_j}\ket{e^{\lz_j}_{\epsilon_j}}
=
\begin{cases}
\prod_{j:\ls_j=\lz_j}\epsilon_j, & \lS\le \lZ\\
0, &\textnormal{otherwise}
\end{cases}.
\end{equation}
Therefore, when $\lS\not\le \lZ$, the identity \eqref{es:key1} holds since both sides vanish. When $\lS\le \lZ$, the right-hand side of \eqref{es:key1} is $\sigma_{\lS}$, while the left-hand side is
\begin{align*}
\sum_{\epsilon}\rho_{\beps,\lZ}\sigma_{\lS}\rho_{\beps,\lZ}
&\stackrel{\eqref{prod}}{=}\sum_{\beps}\prod_{j:\ls_j=\lz_j}\epsilon_j\cdot \ketbra{e^{\lz_1}_{\epsilon_1}}{e^{\lz_1}_{\epsilon_1}}\otimes \cdots\otimes \ketbra{e^{\lz_n}_{\epsilon_n}}{e^{\lz_n}_{\epsilon_n}}\\
&\, \, \,=\sum_{\beps}\left({\textstyle\bigotimes}_{j:\ls_j=0}\ketbra{e^{\lz_j}_{\epsilon_j}}{e^{\lz_j}_{\epsilon_j}}\right)\otimes\left({\textstyle\bigotimes}_{j:\ls_j=\lz_j}\,\epsilon_j
\ketbra{e^{\lz_j}_{\epsilon_j}}{e^{\lz_j}_{\epsilon_j}}\right)\\
&\, \, \,=\left({\textstyle\bigotimes}_{j:\ls_j=0}\sum_{\eps_j}\ketbra{e^{\lz_j}_{\epsilon_j}}{e^{\lz_j}_{\epsilon_j}}\right)\otimes\left({\textstyle\bigotimes}_{j:\ls_j=\lz_j}\sum_{\eps_j}\epsilon_j
\ketbra{e^{\lz_j}_{\epsilon_j}}{e^{\lz_j}_{\epsilon_j}}\right)\\
&\, \, \,=\left(\otimes_{j:\ls_j=0}\un\right)\otimes\left(\otimes_{j:\ls_j=\lz_j}\sigma_{\lz_j}\right)\\
&\, \, \,=\sigma_{\lS}.
\end{align*}
This proves \eqref{es:key1}. The identity \eqref{es:key2} follows from applying \eqref{es:key1} via
$$
\frac{1}{3^n}\sum_{\lZ\in[3]^n}\ce_\lZ(\sigma_{\lS})=\frac{1}{3^n}\sum_{\lZ:\lS\le \lZ}\sigma_{\lS}
$$
and the fact that, for fixed $\lS$, the number of $\lZ\in[3]^n$ satisfying $\lS\le \lZ$ is exactly $3^{n-|\lS|}$.
\end{proof}

%



To some extent, we shall use the conditional expectations $\ce_\lZ$ in the above lemma to reduce the problem to the commutative subalgebras. We will need some tools from the classical setting, which we now recall. 
Any function $f:\{-1,1\}^n\to \C$ has a unique Fourier expansion 
$$
f=\sum_{S\subset [n]}\widehat{f}(S)\chi_S, \qquad \chi_S(x)=\prod_{j\in S}x_j.
$$ 
Thus, $f$ is realized uniquely as a multi-linear (or multi-affine) polynomial, and its degree is defined as $\deg(f):=\max_{\widehat{f}(S)\neq 0}|S|$. Its $k$-homogeneous part is $f_k=\sum_{|S|=k}\widehat{f}(S)\chi_S$.
The following inequality is named after Figiel \cite{MS} and can be found in \cite[Lemma 1]{DMP}: For $f:\{-1,1\}^n\to \R$ of degree at most $d$, its $k$-homogeneous part $f_k$ satisfies 
\begin{equation}
\label{ineq:Rk}
\max_{x\in \{-1,1\}^n}|f_k(x)|\le C(d,k) \max_{x\in \{-1,1\}^n}|f(x)|,\qquad 0\le k\le d,
\end{equation}
 where $C(d,k)$ is a constant depending only on $d$ and $k$, and in particular, $C(d,k)\le (\sqrt{2}+1)^d$.

We will discuss more about \eqref{ineq:Rk} in Section \ref{sec:Figiel}. The next lemma is a qubit analog of \eqref{ineq:Rk}, and we only use the bound $(\sqrt{2}+1)^d$ here, since this is not essential for the proof of Theorem \ref{thm: quantum norm design}. 


\begin{lemma}\label{lem:figiel}
 Let $0\le k\le d\le n$. Suppose that $A$ is a Hermitian operator over $\cH^{\otimes n}$ of degree at most $d$:
 \begin{equation*}
 A=\sum_{|\lS|\le d}\widehat{A}_{\lS}\sigma_{\lS}.
 \end{equation*}
 Then the level $k$-Rademacher projection $\Rad_{k}$ defined by
 \begin{equation*}
\Rad_{k}(A):=\sum_{|\lS|=k}\widehat{A}_{\lS}\sigma_{\lS}
 \end{equation*} 
 satisfies
 \begin{equation}\label{ineq:rk lem}
 \|\Rad_{k}(A)\|\le (\sqrt{2}+1)^d\|A\|.
 \end{equation}
 \end{lemma}
 
 \begin{proof}
 See Section \ref{sec:Figiel}. 
 \end{proof}
 
Now we are ready to prove Theorem \ref{thm: quantum norm design}.


\begin{proof}[Proof of Theorem \ref{thm: quantum norm design}]
We start with the homogeneous case, that is,
\begin{equation}
A=\sum_{\lS:|\lS|=d}\widehat{A}_{\lS}\sigma_{\lS}.
\end{equation}
Fix $\lZ\in[3]^n$ and consider the unitary $U_\lZ:=U_{\lz_1}\otimes \cdots \otimes U_{\lz_n}$ on $\mc H^{\otimes n}$. By Lemma \ref{lem:anti}, for any $\lS\le \lZ$, $U_\lZ^\ast \sigma_{\lS}U_\lZ$ is a tensor product of $\sigma_0$ and $\sigma_3$'s: On the $j$-th place, one has $\sigma_0$ if $\ls_j=0$, and $\sigma_3$ if $\ls_j=\lz_j$. So $U_\lZ^\ast \ce_\lZ(A) U_\lZ$ is a diagonal matrix (or a function on $\{-1,1\}^n$), implying 
\begin{equation}
\| U_\lZ^\ast \ce_\lZ(A) U_\lZ\|=\max_\beps |\tr[U_\lZ^\ast \ce_\lZ(A) U_\lZ\rho_\beps]|
\end{equation}
where $\rho_\beps$, $\beps=(\eps_1,\dots, \eps_n)\in \{-1,1\}^n$, are given by
\begin{equation}
\rho_{\beps}
:= \ket {e^{3}_{\eps_1}}\bra {e^{3}_{\eps_1}}\otimes\dots\otimes  \ket {e^{3}_{\eps_n}}\bra {e^{3}_{\eps_n}}
\stackrel{\eqref{eq:eigenproj}}{=}\Big(\frac12\sigma_0+ \frac12 \eps_1 \sigma_{3}\Big)\otimes\dots\otimes  \Big(\frac12\sigma_0+ \frac12 \eps_n \sigma_{3}\Big).
\end{equation}
Recall that by Lemma \ref{lem:anti}, $U_{\lZ}^\ast \rho_{\beps,\lZ}U_{\lZ}=\rho_{\beps}$. In other words, naively we may diagonalize all the involved $\sigma_{\lS}$ simultaneously to get
\begin{equation}
\label{aS}
\|\ce_\lZ(A)\|
=\max_\beps |\tr[U_\lZ^\ast \ce_\lZ(A) U_\lZ\rho_\beps]|
=\max_\beps |\tr [\ce_\lZ(A) \, \rho_{\beps, \lZ}]|
\end{equation} 
So, there is nothing to do in case  $A=\ce_\lZ(A)$.
If $A$ is not of the particular form $\ce_\lZ(A)$ for some $\lZ \in [3]^n$, we recall \eqref{es:key2 general}
\begin{equation}\label{eq:A-As homo}
A=3^{d-n}\sum_{\lZ\in [3]^n}\ce_\lZ(A)
\end{equation} 
since $A$ is homogeneous of degree $d$. 
Combining \eqref{AsA}, \eqref{aS}, and \eqref{eq:A-As homo}, we have 
\begin{align*}
\label{eq1}
\|A\|
&\stackrel{\eqref{eq:A-As homo}}{\le }
3^{d-n} \sum_{\lZ \in [3]^n}\| \ce_\lZ(A)\|\\
&\stackrel{\eqref{aS}}{=}
3^{d-n}\sum_{\lZ \in [3]^n} \max_{\beps\in \{-1,1\}^n} |\tr [\ce_\lZ(A) \,\rho_{\beps, \lZ}]| \\
& \stackrel{\eqref{AsA}}{=} 3^{d-n}\sum_{\lZ \in [3]^n} \max_{\beps\in \{-1,1\}^n} |\tr [A\, \rho_{\beps, \lZ}]|\\
& \, \, \,\le  \, \, \, 3^d \max_{\lZ, \beps} |\tr[A\,\rho_{\beps, \lZ}]|\,.
\end{align*}
This finishes the proof of the homogeneous case with constant $3^d$.
This recovers the argument of \cite{BGKT} (and \cite[Appendix E]{ACKK}) in our language.

Now let us treat the general case of non-homogeneous $A$ of degree at most $d$:
\begin{equation}
A=\sum_{\lS:|\lS|\le d}\widehat{A}_{\lS}\sigma_{\lS}.
\end{equation}
We are going to follow the same argument as in the homogeneous case, and the main difference is that instead of a nice form of \eqref{eq:A-As homo}, we now only have 
have
\begin{equation}
\label{avd}
A = 3^{-n} \sum_{\lZ \in [3]^n}\sum_{k=0}^{d}3^k\sum_{|\lS|=k, \lS \le \lZ}\widehat{A}_{\lS}\sigma_{\lS}
=3^{-n} \sum_{\lZ \in [3]^n}\sum_{k=0}^{d}3^k\Rad_k[\ce_\lZ(A)]
\end{equation}
by \eqref{es:key2} or \eqref{es:key2 general}.
Then, combining \eqref{AsA},  \eqref{ineq:rk lem}, \eqref{aS}, and \eqref{avd} we obtain
\begin{align*}
\label{final}
\|A\| 
&\stackrel{\eqref{avd}}{\le}  3^{-n} \sum_{\lZ\in [3]^n}\sum_{k=0}^d 3^k\|\Rad_k[\ce_\lZ(A)]\|    \notag
\\
& \stackrel{\eqref{ineq:rk lem}}{\le}  (1+\sqrt{2})^d\,3^{-n} \sum_{\lZ\in [3]^n} \sum_{0\le k\le d} 3^k \|\ce_\lZ(A)\| \notag
\\
&\stackrel{\eqref{aS}}{=} \frac{3^{d+1}-1}{2}(1+\sqrt{2})^d 3^{-n} \sum_{\lZ\in [3]^n}  \max_\eps|\tr[\ce_\lZ(A) \, \rho_{\beps, \lZ}]|\notag
\\
&\stackrel{\eqref{AsA}}{=} \frac{3^{d+1}-1}{2}(1+\sqrt{2})^d 3^{-n} \sum_{\lZ\in [3]^n}  \max_\eps|\tr[A \, \rho_{\beps, \lZ}]|\notag
\\
&\, \, \,\le \,\,\,\frac{3}{2}(3+3\sqrt{2})^d  \max_{\beps, \lZ}|\tr[A \, \rho_{\beps, \lZ}]|.
\end{align*}
This concludes the proof of the non-homogeneous case with constant $\frac{3}{2}(3+3\sqrt{2})^d$.
\end{proof}

\begin{rem}
In the above proof of the non-homogeneous case where we used \eqref{ineq:rk lem} of Lemma \ref{lem:figiel}, we may appeal to its classical version \eqref{ineq:Rk} instead, since we applied \eqref{ineq:rk lem} to $\ce_\lZ(A)$ that lies in the commutative subalgebra $\mathcal{A}_\lZ$. 
\end{rem}


We conclude this section with another application of the above method. The so-called \textit{Bohnenblust--Hille inequality} for discrete hypercubes $\{-1,1\}^n$ states that for any function $f:\{-1,1\}^n\to \R$ of degree at most $d$, we have the dimension-free estimate
\begin{equation}\label{bh boolean}
\|\widehat{f}\|_{\frac{2d}{d+1}}:=\left(\sum_{S}|\widehat{f}(S)|^{\frac{2d}{d+1}}\right)^{\frac{d+1}{2d}}
\le \bh^{\le d}_{\{\pm 1\}}\|f\|_{\infty},
\end{equation}
where $\|f\|_{\infty}$ is the uniform norm over $\{-1,1\}^n$ and $\bh^{\le d}_{\{\pm 1\}}<\infty$ denotes the best constant. We refer to \cite{DMP,DGMS} for more background about this inequality, and the best known estimate is $\bh^{\le d}_{\{\pm 1\}}\le C^{\sqrt{d\log d}}$ with $C>1$ being a universal constant. 

A qubit analog of \eqref{bh boolean} was proved by Huang--Chen--Preskill \cite{HCP} and Volberg--Zhang \cite{VZ}. Namely, for any operator $A$ over $\mc H^{\otimes n}$ of degree at most $d$, one has  
\begin{equation}\label{bh qubit}
\|\widehat{A}\|_{\frac{2d}{d+1}}:=\left(\sum_{\lS}|\widehat{A}_{\lS}|^{\frac{2d}{d+1}}\right)^{\frac{d+1}{2d}}
\le \bh^{\le d}_{M_2}\|A\|,
\end{equation}
where $\bh^{\le d}_{M_2}<\infty$ denotes the best constant. Clearly, $\bh^{\le d}_{\{\pm 1\}}\le \bh^{\le d}_{M_2}$. The main result of \cite{VZ} states that 
\begin{equation}\label{bh 3d}
\bh^{\le d}_{M_2}\le 3^d\bh^{\le d}_{\{\pm 1\}}
\end{equation}
via a reduction method. 

The proof of our Theorem \ref{thm: quantum norm design} relies on a similar reduction idea. In summary, the key ingredient, for $A$ homogeneous of degree $d$ over $\mc H^{\otimes n}$, states that 
\begin{equation}
\max_{\lZ\in[3]^n}\|\ce_\lZ(A)\|\le \|A\|\le 3^{d-n}\sum_{\lZ\in[3]^n}\|\ce_\lZ(A)\|
\le 3^d\max_{\lZ\in[3]^n}\|\ce_\lZ(A)\|.
\end{equation}
Using this idea we can improve upon the previous upper bound \eqref{bh 3d}.

\begin{prop}
For all $d\ge 1$, we have $\bh^{\le d}_{M_2}\le \sqrt{3}^{d+1}\bh^{\le d}_{\{\pm 1\}}$. In other words, for any operator $A$ over $\mc H^{\otimes n}$ of degree at most $d$, one has  
\begin{equation}\label{bh qubit new}
\|\widehat{A}\|_{\frac{2d}{d+1}}\le \sqrt{3}^{d+1}\bh^{\le d}_{\{\pm 1\}}\|A\|.
\end{equation}
\end{prop}

\begin{proof}
Recall that for any $\lZ\in[3]^n$, $\ce_\lZ$ is a conditional expectation, and we have by \eqref{es:key1 general}
\begin{equation*}
\ce_\lZ(A)=\sum_{\lS:\lS\le \lZ}\widehat{A}_{\lS}\sigma_{\lS}.
\end{equation*}
So, applying the Bohnenblust--Hille inequality for the discrete hypercubes \eqref{bh boolean} to $\ce_\lZ(A)$ implies
\begin{equation}
\sum_{\lS:\lS\le \lZ}|\widehat{A}_{\lS}|^{\frac{2d}{d+1}}
\le \left( \bh^{\le d}_{\{\pm 1\}}\|\ce_\lZ(A)\|\right)^{\frac{2d}{d+1}}
\le \left( \bh^{\le d}_{\{\pm 1\}}\|A\| \right)^{\frac{2d}{d+1}}.
\end{equation}
Summing over all $\lZ\in[3]^n$, and using the fact that $|\{\lZ:\lS\le \lZ\}|=3^{n-|\lS|}$ for any fixed $\lS$, we get
\begin{equation}
3^{n-d}\sum_{\lS:|\lS|\le d}|\widehat{A}_{\lS}|^{\frac{2d}{d+1}}
\le \sum_{\lS:|\lS|\le d}3^{n-|\lS|}|\widehat{A}_{\lS}|^{\frac{2d}{d+1}}
=\sum_{\lZ\in [3]^n}\sum_{\lS:\lS\le \lZ}|\widehat{A}_{\lS}|^{\frac{2d}{d+1}}
\le 3^n\left( \bh^{\le d}_{\{\pm 1\}}\|A\|\right)^{\frac{2d}{d+1}} .
\end{equation}
Rearranging, this is exactly \eqref{bh qubit new}.
\end{proof}

The above method also provides simple bounds on operator norms of \emph{random Hamiltonians}. We omit the details here since we will obtain some better bounds using a different method in Appendix \ref{appendix b}.

\section{The cardinality of quantum norm designs}
\label{card}

According to Theorem \ref{thm: quantum norm design}, the sets $\mathbf{X}_n$ of the quantum norm design can be chosen to be of exponential size $|D^{\otimes n}|=6^n$.
In this section we construct universal sampling sets of smaller size, at the cost of increasing the norm design constant.

\begin{theorem}\label{thm: small quantum norm design}
Fix $d\ge 1$ and $\varepsilon > 0$.
Then there exists $C = C(\varepsilon) > 0$ and a norm design $\X = {\X_1,\X_2,\ldots}$ such that each $\X_n$ is a set of product states and has cardinality at most $C\cdot (1 + \varepsilon)^n$.
\end{theorem}

The bound on the size in Theorem \ref{thm: small quantum norm design} is essentially optimal.

\begin{theorem}
    \label{thm: lower bound design}
    For every $C > 0$ there exists $\varepsilon = \varepsilon(C)$ such that the following holds.
    Suppose that $\mathbf{Y}_n$ is a set of 
    states on $\mathcal{H}^{\otimes n}$ such that for any operator $A$ on $\mathcal{H}^{\otimes n}$ of degree $1$, we have
    \begin{equation}\label{ineq:small design2}
	\sup_{\rho\in \Y_n}|\tr[A\rho]|
	\le \|A\|
	\le C	\sup_{\rho\in \Y_n}|\tr[A\rho]|.
	\end{equation}
    Then 
    $$
        |\mathbf{Y}_n| \ge (1 + \varepsilon)^n\,.
    $$
\end{theorem}

We now proceed to the proof of Theorem \ref{thm: small quantum norm design}; Theorem \ref{thm: lower bound design} will be proved afterwards.

Theorem \ref{thm: small quantum norm design} follows from a reduction to commutative polynomials, and then the main theorem of \cite{BKSVZ}, which is a discretization inequality for commutative polynomials. For the convenience of the reader, we state here the special case of this theorem that we will be using. Let $\D=\{z:|z|<1\}$ be the open unit disk.  

\begin{prop}
    \label{prop:smalldes}
    Fix $d\ge 1$ and $-1<a<b<1$.
    For every $\varepsilon > 0$ there exists $C = C(\varepsilon) > 0$ such that the following holds. For every $n\ge 1$ there exists a set $\mathbf{S}_n \subset \{a,b\}^n$ with $|\mathbf{S}_n| \le C\cdot (1+\varepsilon)^n$, such that for every multi-affine analytic polynomial $f:\D^{n}\to \C$ of degree at most $d$, we have 
    \begin{equation}
    \sup_{\z\in \D^{n}}|f(\z)|\le C(d)\sup_{\z\in \mathbf{S}_n}|f(\z)|.
    \end{equation} 
\end{prop}

Let $A$ be Hermitian of degree at most $d$. We use the following notation introduced in Section \ref{sec:multilinear}. Let $\lZ\in[3]^n$. For 
\begin{equation}\label{eq:A}
A=\sum_{\lS:|\lS|\le d}\widehat{A}_{\lS}\sigma_{\lS}
\end{equation}
we write 
\begin{equation}
\ce_\lZ(A)=\sum_{|\lS|\le d, \lS\le \lZ}\widehat{A}_{\lS}\sigma_{\lS}.
\end{equation}
Then we have shown in \eqref{aS} and \eqref{AsA} that 

\begin{equation}\label{lem product}
\|\ce_\lZ(A)\| 
=\max_{\beps,\lZ}|\tr[A\rho_{\beps,\lZ}]|
\le \max_{\substack{\psi = \psi_1 \otimes \dotsb \otimes \psi_n\\ \|\psi\| = 1}} |\langle \psi | A | \psi \rangle|.
\end{equation}

For any vector $\x\in \R^{3n}$ of the form
\[
    \x = \left(x^{(1)}_1, x^{(1)}_2, x^{(1)}_3, x^{(2)}_1, x^{(2)}_2, x^{(2)}_3,\dotsc, x^{(n)}_1, x^{(n)}_2, x^{(n)}_3\right),
\]
we will denote
\begin{equation}
    \label{eq: rho_a}
    \rho(\x) = \rho(x^{(1)}) \otimes\dotsb\otimes \rho(x^{(n)})
\end{equation}
where each $\rho(x^{(j)})$ is given by
\begin{equation}
    \label{eq: rho_a2}
    \rho(x^{(j)}) = \frac{1}{2}(\sigma_0 + x_1^{(j)} \sigma_1 + x_2^{(j)} \sigma_2 + x_3^{(j)} \sigma_3).
\end{equation}
Any $A$ as in \eqref{eq:A} corresponds to a classical polynomial 
\begin{equation}
p_A(\x):=\sum_{\lS:|\lS|\le d}\widehat{A}_{\lS}\prod_{j:\ls_j\neq 0}x^{(j)}_{\ls_j}
\end{equation}
of the same degree. Moreover, for all $\x\in \R^{3n}$,
\begin{equation}\label{eq:lem polynomial}
\tr[A\rho(\x)]=p_A(\x).
\end{equation}
In fact, the equation is linear in $A$, so it suffices to verify it for $A=\sigma_{\lS}$. Then by orthogonality, 
$$
\tr[A\rho(\x)]=\prod_{j}\tr[\sigma_{\ls_j}\rho(x^{(j)})]
=\prod_{j:\ls_j\neq 0}x_{\ls_j}^{(j)}=p_A(\x).
$$

Now, we are ready to reduce the problem of finding small quantum-grids to the problem of finding small grids for polynomials.

\begin{proof}[Proof of Theorem \ref{thm: small quantum norm design}]
Recalling \eqref{avd}:
\begin{equation}
A = 3^{-n} \sum_{\lZ\in[3]^n}\sum_{k=0}^{d}3^k\sum_{|\lS|=k, \lS \le \lZ}\widehat{A}_{\lS}\sigma_{\lS},
\end{equation}
which implies  (let $\tilde A:=\sum_{k=0}^{d}3^k\sum_{|\lS|=k}\widehat{A}_{\lS}\sigma_{\lS}$ )
\begin{equation}
\|A\|
\le \max_{\lZ\in[3]^n}\left\|\sum_{k=0}^{d}3^k\sum_{|\lS|=k, \lS \le \lZ}\widehat{A}_{\lS}\sigma_{\lS} \right\|= \max_{\lZ\in[3]^n}\|  \ce_{\lZ}(\tilde A)\|
\end{equation}
Applying \eqref{lem product} to $\tilde A=\sum_{k=0}^{d}3^k\sum_{|\lS|=k}\widehat{A}_{\lS}\sigma_{\lS}$ instead of $A$ and using \eqref{AsA} that says  $\tr [\ce_{\lZ}(\tilde A)\, \rho_{\eps, \lZ}] = \tr [\tilde A  \rho_{\eps, \lZ}]$, we see that the norm inside the maximum is bounded by (as for any fixed $\lZ$ the state $ \rho_{\eps, \lZ}]$ is a particular case of $\rho(\x)$)
$$
\sup_{\x \in \R^{3n} : \max_{1 \le j \le n} \|x^{(j)}\|_2 \le 1} \left|\sum_{k=0}^{d}3^k\sum_{|\lS|=k}\widehat{A}_{\lS}\tr[\sigma_{\lS} \rho(\x)]\right|
$$
where in the sup, $\|x^{(j)}\|_2$ denotes the $\ell^2$ norm of $x^{(j)}=(x^{(j)}_1,x^{(j)}_2,x^{(j)}_3)$ and the bound of 1 follows because $x^{(j)}$ is a Bloch vector. Following the computation in verifying \eqref{eq:lem polynomial}, one has
$$
\sum_{k=0}^{d}3^k\sum_{|\lS|=k}\widehat{A}_{\lS}\tr[\sigma_{\lS} \rho(\x)]
=\sum_{k=0}^{d}\sum_{|\lS|=k}\widehat{A}_{\lS}\prod_{j:\ls_j\neq 0}3x_{\ls_j}^{(j)}
=p_A(3\x).
$$
All combined, we have shown
\begin{equation}\label{ineq:A  PA}
   \|A\| \le \sup_{\x \in \R^{3n} : \max_{1 \le j \le n} \|x^{(j)}\|_2 \le 1} |p_A(3\x)|.
\end{equation}
Consider the polynomial $q_A(\x) := p_A(3\x)$ that has the degree at most $d$, and note that 
\[
    \left\{\x \in \R^{3n}\ : \max_{1 \le j \le n} \|x^{(j)}\|_2 \le 1\right\} \subset \D^{3n},
\]
so
\[
    \|A\| \le \sup_{\x\in \D^{3n}} |q_A(\x)|.
\]
By Proposition \ref{prop:smalldes} applied to $\{a,b\} = \{-1/6,1/6\}$, there exist universal sets
\[
    \mathbf{S}_n \subset \left\{-\frac{1}{6}, \frac{1}{6}\right\}^n
\]
of size 
\[
    \lvert \mathbf{S}_n \rvert  \le C(d, \epsilon) (1 + \epsilon)^n
\]
and a constant $C(d)$ such that for the above $q_A$ 
\[
    \sup_{\D^{3n}} |q_A| \le C(d) \sup_{\mathbf{S}_n} |q_A|.
\]
This, together with \eqref{ineq:A  PA}, implies 
\[
    \|A\| \le C(d) \sup_{\x \in \mathbf{S}_n} |q_A(\x)| = C(d) \sup_{\x \in \mathbf{S}_n} |p_A(3\x)| 
    = C(d) \sup_{\x' \in 3\mathbf{S}_n} |p_A(\x')|.
\]
To relate this back to $A$ we use again \eqref{eq:lem polynomial}
\[
   C(d) \sup_{\x' \in 3\mathbf{S}_n} |p_A(\x')| = C(d) \sup_{\x' \in 3\mathbf{S}_n} \tr[A \rho(\x')].
\]
This completes the proof, since for all $\x' \in 3\mathbf{S}_n\subset \{-1/2,1/2\}^n$, $\rho(\x')$ is a quantum state.
\end{proof}
%
%

Now we turn to the lower bound for the number of necessary sample states. It follows very similarly to the analogous lower bound for commutative polynomials, which is proved in \cite{BKSVZ}.
We reproduce the argument here again for the convenience of the reader.

\begin{proof}[Proof of Theorem \ref{thm: lower bound design}]
We consider the degree $1$ matrix polynomials
$$
    A(\x) = \sum_{j = 1}^n x_j \sigma_j^{(3)}\,,\quad\quad \x \in \{-1,1\}^n\,,
$$
where $\sigma_j^{(3)}$ is the tensor product of $n$ copies of $\sigma_0=\un$, except for the $j$-th place where we have $\sigma_3$. These are $2^n$ diagonal matrices and the eigenvalues of $A(\x)$ are  
$$
    \sum_{i = 1}^n \epsilon_i x_i, \qquad \epsilon_1,\dots,\epsilon_n\in \{-1,1\}.
$$
In particular, for all $\x \in \{-1,1\}^n$
\[
    \|A(\x)\| = \sup_{\epsilon \in \{-1,1\}^n} \Big| \sum_{i = 1}^n \epsilon_i x_i \Big| = n.
\]

Now suppose that a constant $C > 0$ and a set of $n$-qubit states $\mathbf{Y}_n$ are given such that \eqref{ineq:small design2} holds. For each $\rho \in \mathbf{Y}_n$, we consider the set $\mathbf{H}(\rho)$ of all $\x \in \{-1,1\}^n$ such that 
\begin{equation}
    \label{eq: Hrho}
     n = \|A(\x)\| \le C \lvert\tr[A(\x)\rho]\rvert\,.
\end{equation}
Expanding $\rho$ in the Pauli basis, we write
\[
    \rho = 2^{-n} \sum_{\mathbf{\lS} \in \{0,1,2,3\}^n} \widehat{\rho}_\mathbf{\lS} \sigma_{\mathbf{\lS}}, 
\]
and since $\rho$ is a state, we have for $\lS \ne (0,\dots, 0)$
\begin{equation}
    \label{eq norm sum}
    |\widehat{\rho}_\mathbf{\lS}| \le 1.
\end{equation}
Then \eqref{eq: Hrho} is equivalent to
$$
     \frac{n}{C} \le \left|\sum_{j = 1}^n x_j  \widehat{\rho}_{3\mathbf{e_j}} \right|\,,
$$
where $3\mathbf{e_j}$ is the multiindex that is $0$ in all places except the $j$-th, where it is $3$.
By Hoeffding's inequality, we can thus estimate the number of $\x \in \{-1,1\}^n$ satisfying \eqref{eq: Hrho} by
\begin{equation}
    \label{eq: Hbound}
    |\mathbf{H}(\rho)| = 2^n \Pr_{\x} \left[ \left|\sum_{j=1}^n x_j \widehat{\rho}_{3\mathbf{e_j}} \right|\ge \frac{n}{C}\right]
    \le 2^n \exp\bigg(-\frac{1}{2} \frac{n^2}{C^2\sum_{j=1}^n |\widehat{\rho}_{3\mathbf{e_j}}|^2}\bigg)\,.
\end{equation}
From \eqref{eq norm sum}, we conclude that
\begin{equation}
    \label{eq: rhobound}
    \sum_{j=1}^n |\widehat{\rho}_{3\mathbf{e_j}}|^2 \le n\,.
\end{equation}
Combining \eqref{eq: Hbound} and \eqref{eq: rhobound}, we find
$$
    |\mathbf{H}(\rho)| \le 2^n \exp\left(-\frac{n}{2C^2}\right) .
$$
Our assumption, that \eqref{ineq:small design2} holds, is equivalent to the statement that $\{-1,1\}^n$ is contained in the union of the sets $\mathbf{H}(\rho)$, $\rho \in \mathbf{Y}_n$.
It follows that 
$$
    2^n = |\{-1,1\}^n| \le \sum_{\rho \in \mathbf{Y}_n} |\mathbf{H}(\rho)| \le |\mathbf{Y}_n| 2^n \left( \exp\left(-\frac{1}{2C^2}\right)\right)^n\,,
$$
so 
$$
    |\mathbf{Y}_n| \ge \exp\left(\frac{1}{2C^2}\right)^n = (1 + \varepsilon)^n\,,
$$
where $\varepsilon = \varepsilon(C) = \exp\left(\frac{1}{2C^2}\right) - 1$.
\end{proof}

\section{Geometry of norm designs: norm designs from any 1-qubit 2-design}
\label{sec:2design}
Theorem \ref{thm: quantum norm design} establishes the grid of Pauli eigenstates as a quantum norm design.
Single-qubit Pauli eigenstates also form a quantum 2-design (actually a 3-design).
Here we demonstrate that an $n$-fold tensor power of any $2$-design is a quantum norm design.

%
	
	\medskip

Recall that a quantum 1-qubit 2-design is a set $D$ of 1-qubit states such that a certain matrix quadrature formula is satisfied:
	\begin{equation}
	\label{2d}
	\int_{\ket{\psi}\sim\mathrm{Haar}(\C^2)}\psi^{\otimes 2}\,\mathrm{d}\psi  =\frac1{|D|} \sum_{\ket{\psi}\in D} \psi^{\otimes 2}\,.\
	\end{equation}
	where Haar($\C^2$) denotes the uniform probability measure on 1-qubit pure states and $\psi$ refers to the rank one projection onto $\ket{\psi}$, or $\psi=\ket{\psi}\bra{\psi}$.
	There are many such collections $D$, but the smallest is
	one where $|D|=4$: 
	\begin{align*}
	\ket{\psi_1} &=\ket 0,\qquad & \ket{\psi_2} &= \frac1{\sqrt{3}} \ket 0+ \sqrt{\frac23} \ket 1, \\
	\ket{\psi_3}&=  \frac1{\sqrt{3}} \ket 0 + \sqrt{\frac23}e^{i\frac{2\pi }{3}} \ket 1, \quad & \ket{\psi_4} &= \frac1{\sqrt{3}} \ket 0 + \sqrt{\frac23}e^{i\frac{4\pi }{3}} \ket 1\,.
	\end{align*}

	We refer the reader to \textit{e.g.} \cite{AE} for an introduction to quantum $t$-designs.

	\begin{theorem}
	\label{2design-th}
		Let $D$ be any 1-qubit 2-design.
		Then for all d-local Hamiltonians $H$ we have
		\[
		\|H\|
		\leq C_d\max_{\ket{\psi}\in D^{\otimes n}}|\tr[H \psi]|\,.
		\]
		Here $C_d$ is a universal constant depending on $d$ only, which can be taken to be $C\cdot 3^{d^2}$. 
	\end{theorem}
	
	In the proof of Theorem \ref{thm: quantum norm design} we took essential advantage of the geometry of our chosen $\mathbf{X}_n$ to reduce to commutative subalgebras.
	In the setting of Theorem \ref{2design-th}, where we have much less control over the geometry, there does not seem to be a similar reduction.
	Instead, we find that a certain polynomial of depolarizing channels can take the place of Rademacher projection at the expense of a worse dependence on $d$ in the dimension-free constant.
		
	\begin{proof}
	Consider the 1-qubit depolarizing channel with parameter $1/3$, which has the following integral formulation.
	With $M$ any 1-qubit operator,
	\[
	\mc N(M) = 2\int_{\ket{\psi}\sim\mathrm{Haar}({\C}^2)}\tr[M\psi]\psi\,\mathrm{d}\psi\,
	.\]
	Note that $\mc N$ acts on the Pauli matrices as
	\begin{equation}
	\label{eq:haar-int}
		 \mc N(\sigma_0)=\mc N(I)=I \qquad\text{and}\qquad \mc N(\sigma_j)=\frac{\sigma_j}3, \quad j=1,2,3.
	\end{equation}

	
	\medskip
	
	Put $\mc E = \mc N^{\otimes n}$ and for any $n$-qubit operator $A$ let $A_\ell$ be the $\ell$-homogeneous part of $A$.
	Then
	\[\mc E^{k}(A):=\underbrace{\vphantom{a_b}\mc E\circ\cdots \circ\mc E}_{k\text{ times}}(A)=\sum_\ell \left(\frac13\right)^{\ell\cdot k}A_\ell\,.\]	
	Let $c=(c_1,\ldots, c_{d+1})$ be the solution to the following Vandermonde system,
	\[
	\begin{pmatrix}
		1 & 1 & \cdots & 1\\
		3^{-1} & 3^{-2} & \cdots & 3^{-(d+1)}\\
		3^{-2} & 3^{-4} & \cdots & 3^{-2(d+1)}\\
		\vdots & \vdots & \ddots & \vdots\\
		3^{-d} & 3^{-2d} & \cdots & 3^{-d (d+1)}
	\end{pmatrix}
	\begin{pmatrix}
		c_1\\
		c_2\\
		\vdots\\
		c_{d+1}
	\end{pmatrix} = \begin{pmatrix}
		1\\
		1\\
		\vdots\\
		1
	\end{pmatrix}.
	\]
	Consider an $n$-qubit (mixed) state $\rho$.
	Then for Hamiltonian $H$ of degree at most $d$, we have
	\begin{equation}
	\label{eq:vandermonde-poly}
		H=\sum_{k=1}^{d+1}c_k \mc E^{k}(H).
	\end{equation}
		
	Coming back to any $2$-design $D$, recall that by the $2$-design property of $D$ we have for any 1-qubit operator $M$
	\[
	\mc N(M)=\frac{2}{|D|}\textstyle\sum_{\psi\in D}\tr[M\psi]\psi
	\]
	 and thus for any $n$-qubit operator $A$
	\[
	\mc E(A)=\frac{2^{n}}{|D|^{n}}\sum_{\psi_1\in D^{\otimes n}}\tr[A \psi_1]\psi_1.
	\]
	 Thus for all $k\ge 1$ and for any $n$-qubit operator $A$
	\[
	\mc E^{k}(A)=\frac{2^{kn}}{|D|^{kn}}\sum_{\psi_1,\ldots, \psi_{k}\in D^{\otimes n}}\tr[A \psi_1]\tr[\psi_1 \psi_2]\cdots\tr[\psi_{k-1}\psi_k]\psi_k.
\]
%
	We combine this observation with \eqref{eq:vandermonde-poly} to estimate
	\begin{align*}
		\|H\|&=\max_{\ket{\varphi}}|\tr[H\varphi]|\\
		&=\max_\rho\left|\tr\big[\textstyle\sum_{k=1}^{d+1}c_k\,\mc E^{k}(H)\rho\big]\right|\\
		&\leq \max_\rho\sum_{k}|c_k|\sum_{\psi_1,\ldots, \psi_{k}\in D^{\otimes n}}\big|\tr[H\psi_1]\big|\frac{\tr[\psi_1\psi_2]}{(|D|/2)^n}\cdots\frac{\tr[\psi_{k-1}\psi_k]}{(|D|/2)^n}\frac{\tr[\psi_k\rho]}{(|D|/2)^n}\\
		&\leq \max_{\rho}\left(\max_{\psi_1\in D^{\otimes n}}|\tr[H\, \psi_1]|\right)\sum_{k}|c_k|\sum_{\psi_1,\ldots, \psi_{k}\in D^{\otimes n}}\frac{\tr[\psi_1\psi_2]}{(|D|/2)^n}\cdots\frac{\tr[\psi_k\rho]}{(|D|/2)^n}\\
		&= \|c\|_1\max_{\psi\in D^n}|\tr[H\, \psi]|\,.
	\end{align*}	
	where $\rho$ is any $n$-qubit state. In the last line we used that 
	\[
	1=\tr[\ce^k(\un^{\otimes n})\rho]=\frac{2^{kn}}{|D|^{kn}}\sum_{\psi_1,\ldots, \psi_{k}\in D^{\otimes n}}\tr[\psi_1 \psi_2]\cdots\tr[\psi_{k-1}\psi_k]\tr[\psi_k\rho].
	\]
	\medskip	
	It remains to estimate $\|c\|_1$, which is at most $\big[\prod_{1\le j <k\le d+1} |\lambda_j-\lambda_k|\big]^{-1}$ for $\lambda_j=3^{-j}$, that is $\le C\, 3^{d^2}$.
	
	\end{proof}

	An advantage of choosing the 1-qubit $2$-design listed above and consisting of just $4$ elements is that now we got the grid of product states having cardinality $4^n$.

\section{Figiel's estimate for level-$k$ Rademacher projections and other related inequalities}
\label{sec:Figiel}
In this section, we will discuss Lemma \ref{lem:figiel} in more detail. The inequality is named after Figiel, and we are going to prove a qubit version of it. Recall that for any $0\le k\le n$, the level $k$-Rademacher projection $\Rad_{k}$ is a linear operator given by 
 \begin{equation*}
 \Rad_{k}(A)=\sum_{|\lS|=k}\widehat{A}_{\lS}\sigma_{\lS}
 \end{equation*}
 for any operator $A$ over $\cH^{\otimes n}$. We recall Lemma \ref{lem:figiel} with more details. 
 
 \begin{prop}\label{thm:figiel}
 Let $0\le k\le d\le n$. Suppose that $A$ is an operator over $\cH^{\otimes n}$ of degree at most $d$:
 \begin{equation*}
 A=\sum_{|\lS|\le d}\widehat{A}_{\lS}\sigma_{\lS}.
 \end{equation*}
 Then the level $k$-Rademacher projection $\Rad_{k}$ satisfies
 \begin{equation}
 \|\Rad_{k}(A)\|\le C(d,k)\|A\|,
 \end{equation}
 where $C(d,k)$ is a constant depending only on $d$ and $k$. Moreover, $C(d,k)$ is the same constant as in the discrete hypercube case \eqref{ineq:Rk}, and in particular, $C(d,k)\le (\sqrt{2}+1)^d$.
 \end{prop}

\begin{rem}
The constant $C(d,k)$ is given in terms of coefficients $d$-th Chebyshev polynomial of the first kind which satisfies a better estimate $C(d,k)\le\frac{d^k}{k!}$. The constant $\sqrt{2}+1$ in $C(d,k)\le (\sqrt{2}+1)^d$ is best possible.
\end{rem} 
 
To prove Theorem \ref{thm:figiel}, we follow the argument in \cite{EI3}. For any $A$ over $\cH^{\otimes n}$ of the form
\begin{equation*}
    A=\sum_{\lS\in \{0,1,2,3\}^n}\widehat{A}_{\lS}\sigma_{\lS},
\end{equation*}
consider the family of linear operators
\begin{equation}
P_r(A)=\sum_{\lS}r^{|\lS|}\widehat{A}_{\lS}\sigma_{\lS},\qquad r\in [-1,1].
\end{equation}
It is well-known that $P_r$ is a contraction over all Schatten-$p$ classes, $p\in [1,\infty]$, when $r\in [0,1]$. For any $p\in [1,\infty]$, we denote $\|A\|_p$ the Schatten-$p$ norm of $A$, and when $p=\infty$, $\|A\|_{\infty}=\|A\|$ is the operator norm. The following lemma says more about it. 


\begin{lemma}
\label{const1}
For any operator $A$ over $\cH^{\otimes n}$, we have for all $p\in [1,\infty]$ that
\begin{equation}\label{contraction}
\|P_r(A)\|_p\le \|A\|_p,\qquad r\in [-1,1].
\end{equation}
\end{lemma} 

\begin{proof}

The map $P_r$ is the $n$-fold tensor product of the map 
$$
B  \mapsto r\,B + (1-r)\cdot 2^{-1}\tr[B]\un
$$
over $2$-by-$2$ complex matrix algebra that is completely positive when $r\in [0,1]$. So, $P_r$ is (completely) positive and by Russo--Dye theorem \cite[Theorem 2.3.7]{B} (see also \cite{RD}), $\|P_r(A)\|\le \|P_r(\un)\|\|A\|=\|A\|$, since $P_r$ is unital. Note that $P_r$ is also trace-preserving, so it is also a contraction in $\|\cdot\|_1$.  Then by complex interpolation, $P_r$ is a contraction in $\|\cdot\|_p$ for all $p\in [1,\infty]$ when $r\in [0,1]$.

To prove \eqref{contraction} for $r\in [-1,0)$, note that it suffices to show it for $r=-1$, since $P_r=P_{-r} P_{-1}$ would be a composition of two contractions $P_{-1}$ and $P_{-r},-r\in (0,1]$. 

In order to prove $\eqref{contraction}$ for $r=-1$, note that 
    $$
    \sigma_2^3=\sigma_2,\qquad \textnormal{while}\qquad \sigma_2 \sigma_j \sigma_2=-\sigma_j, \qquad j=1,3
    $$
   and
    $$
    \sigma_2^T=-\sigma_2, \qquad \textnormal{while}\qquad \sigma_j^T=\sigma_j,\qquad j=1,3.
    $$
   Here, $A^T$ denotes the transpose of $A$.
    Thus \begin{equation}\label{eq:rotation}
        (\sigma_2 \sigma_j \sigma_2)^T= -\sigma_j, \qquad j=1,2,3.
    \end{equation}
    This, together with $(\sigma_2 \sigma_0 \sigma_2)^T= \sigma_0$, implies 
    \begin{equation}\label{eq:flip}
    (UAU)^T=\sum_{\lS}(-1)^{|\lS|}\widehat{A}_{\lS}\sigma_{\lS}=P_{-1}(A)
    \end{equation}
where $U:=\sigma_2\otimes \cdots \otimes \sigma_2$ is an Hermitian unitary. Therefore, we have 
\begin{equation}\label{isometry}
\|P_{-1}(A)\|_p=\|(UAU)^T\|_p=\|UAU\|_p=\|A\|_p.
\end{equation}
where in the second equality we used the fact that the transpose preserves the Schatten-$p$ norms.
\end{proof}
%

\begin{proof}[Proof of Theorem \ref{thm:figiel}]
    For any operator $B$ over $\mc H^{\otimes n}$ with $\|B\|_1\le 1$, consider $p(r):=\langle P_r(A), B\rangle$ with $P_r$ as above. Then $p$ is a polynomial of degree at most $d$, and its $k$-homogeneous part is $\langle \Rad_k(A),B\rangle$. So, by classical Figiel's inequality \eqref{ineq:Rk}
    $$
    |\langle \Rad_k(A),B\rangle|
    \le C(d,k)\sup_{[-1,1]}|p|.
    $$
    By H\"older's inequality and Lemma \ref{const1}, we have 
    $$
    \sup_{[-1,1]}|p|
    \le \sup_{r\in [-1,1]}\|P_r(A)\|\cdot \|B\|_1
    \le \|A\|.
    $$
    Therefore, 
    $$
\|\Rad_k(A)\|=\sup_{\|B\|_1\le 1} |\langle \Rad_k(A),B\rangle|\le C(d,k)\|A\|.
    $$
    This finishes the proof. 
 \end{proof}
 
 More consequences follow from Lemma \ref{const1}, and we present here one of them as an example. 

\begin{prop}
Let $A$ be any operator over $\mc H^{\otimes n}$ of degree at most $d$. Then for all $p\in [1,\infty]$ we have 
\begin{equation}
\|P_r(A)\|_p\ge \frac{1}{T_d(1/r)}\|A\|_p, \qquad r\in [0,1],
\end{equation}
where $T_d$ is the $d$-th Chebyshev polynomial of the first kind.
\end{prop}

\begin{proof}
The proof is the same as in \cite{EI2}. In fact, according to the proof of \cite[Theorem 1]{EI2}, for any $r\in [0,1]$ there exists a complex measure $\mu_r$ on $[-1,1]$ such that
$$
\int_{-1}^{1}x^k\textnormal{d}\mu_r(x)=r^{-k}, \qquad k=0,1,\dots, d
$$
and $\|\mu_r\|\le T_d(1/r)$. Here, $\|\mu\|$ denotes the total variation norm of a complex measure $\mu$. Thus for $A$ of degree at most $d$:
$$
P_{1/r}(A)
=\sum_{|\lS|\le d}r^{-|\lS|}\widehat{A}_{\lS}\sigma_{\lS}
=\int^{1}_{-1}\sum_{|\lS|\le d}x^{|\lS|}\widehat{A}_{\lS}\sigma_{\lS}\textnormal{d}\mu_{r}(x)
=\int^{1}_{-1}P_x(A)\textnormal{d}\mu_r(x).
$$ 
This, together with Lemma \ref{const1} and the triangle inequality, implies
 $$
 \|A\|_p
 = \|P_{1/r}P_{r}(A)\|_p
 = \|\int^{1}_{-1}P_x(P_r(A))\textnormal{d}\mu_r(x)\|_p
 \le \int^{1}_{-1}\|P_x(P_r(A))\|_p\textnormal{d}|\mu_r|(x)
 \le \|\mu_r\|\cdot \|P_r(A)\|_p
 $$
which concludes the proof because $\|\mu_r\|\le T_d(1/r)$.
\end{proof}


\section{Constant $9$ for $2$-local Hamiltonians}
 \label{appendix}
Recall that for general $d$-local Hamiltonians, our approximation constant for a small norm design can be chosen to be $\frac{3}{2}(3+3\sqrt{2})^d$, and if $A$ is further homogeneous, one can improve the constant to $3^d$. When $d=2$, Lieb already proved a similar result for homogeneous Hamiltonian in \cite{L} with a constant $9=3^2$. In case it is non-homogeneous (and traceless), Bravyi--Gosset--K\"onig--Temme \cite{BGKT} obtained the same constant $9$ using a beautiful observation to reduce the problem to the homogeneous case, which we shall explain below.

Let $A=A_1+A_2$ be a traceless self-adjoint operator on $\cH^{\otimes n}$, where $A_k,k=1,2$ are the $k$-homogeneous parts of $A$, respectively. 
Bravyi--Gosset--K\"onig--Temme considered the operator
$$
A':= A_2 \otimes \sigma_0 + A_1\otimes \sigma_3
=\begin{pmatrix}
A_2+A_1 &0\\
0& A_2-A_1
\end{pmatrix}.
$$
which is homogeneous of degree $2$ over $\cH^{\otimes (n+1)}$.

Moreover, one has
\begin{equation}\label{eq:reduction}
\|A'\|=\|A\|
\end{equation}
so that one can reduce the problem to the homogeneous setting. In fact, recall that $P_{-1}(A)=A_2-A_1$, so \eqref{eq:reduction} follows from \eqref{isometry}
$$
\|A'\|=\max\{ \|A_2+A_1\|, \|A_2-A_1\|\}=\max\{\|A\|,\|P_{-1}(A)\|\}=\|A\|.
$$

To conclude the proof of constant $9$ for $A=A_1+A_2$, it suffices to apply our results for homogeneous $A'$. More precisely, let $S$ be the collection of all maps $s:[n]\to [3]$ as before, and $S'$ the collection of all maps $s':[n+1]\to [3]$. We use $\epsilon$ to denote any vector in $\{-1,1\}^n$, and we shall use $\epsilon'$ for any vector in  $\{-1,1\}^{n+1}$. Recall that $\rho_{\epsilon',s'}$ is a state of the form
$$
\rho_{\eps', s'}
=\ketbra{e^{s'(1)}_{\eps'_1}}{e^{s'(1)}_{\eps'_1}}\otimes\dots\otimes  \ketbra{e^{s'(n)}_{\eps'_n}}{e^{s'(n)}_{\eps'_n}}\otimes \ketbra{e^{s'(n+1)}_{\eps'_{n+1}}}{e^{s'(n+1)}_{\eps'_{n+1}}}.
$$

Then, combining \eqref{eq:reduction} and our proof of Theorem \ref{thm: quantum norm design} in the homogeneous 
case:
\begin{align*}
\|A\|
=\|A'\|
\le 9\max_{s',\epsilon'}|\tr[A'\rho_{\epsilon',s'}]|
=9\max_{s',\epsilon'}|\tr[A_1\rho_{\epsilon,s}]-\delta_{s'(n+1),3}\epsilon'_{n+1}\tr[A_2\rho_{\epsilon,s}]|,
\end{align*}
where $\epsilon=(\epsilon'_1,\dots, \epsilon'_n)$ and $s=s'|_{[n]}$. By definition, we have 
\begin{align*}
&\max_{s'(n+1)\in [3],\epsilon'_{n+1}=\pm 1}|\tr[A_1\rho_{\epsilon,s}]-\delta_{s'(n+1),3}\epsilon'_{n+1}\tr[A_2\rho_{\epsilon,s}]|\\
&\qquad=\max\left\{ 
|\tr[A_1\rho_{\epsilon,s}]|,
|\tr[A_1\rho_{\epsilon,s}]-\tr[A_2\rho_{\epsilon,s}]|,
|\tr[A_1\rho_{\epsilon,s}]+\tr[A_2\rho_{\epsilon,s}]|
\right\}\\
&\qquad =\max\left\{ 
|\tr[(A_2-A_1)\rho_{\epsilon,s}]|,
|\tr[(A_2+A_1)\rho_{\epsilon,s}]|
\right\}.
\end{align*}

Now we make one observation before taking the maximum over the rest $s,\epsilon$. Recall that
$A_2-A_1=(U(A_2+A_1)U)^T$, so 
$$
\tr[(A_2-A_1)\rho_{\epsilon,s}]
=\tr[(U(A_2+A_1)U)^T\rho_{\epsilon,s}]
=\tr[(A_2+A_1)U\rho_{\epsilon,s}^TU].
$$
Recalling \eqref{r-s} 
$$
\rho_{\eps, s}=\bigotimes_{j=1}^{n} \Big(\frac12\sigma_0+ \frac12 \eps_j \sigma_{s(j)}\Big),
$$
and \eqref{eq:rotation}
$$
\sigma_2\sigma_j^T\sigma_2=-\sigma_j, \qquad j=1,2,3,
$$
we have
$$
U^\ast\rho_{\eps, s}^TU= \bigotimes_{j=1}^{n} \Big(\frac12\sigma_0+ \frac12 \eps_j\sigma_2 \sigma_{s(j)}^T\sigma_2\Big)
=\bigotimes_{j=1}^{n} \Big(\frac12\sigma_0- \frac12 \eps_j \sigma_{s(j)}\Big)
=\rho_{-\eps, s}.
$$
Here, $-\epsilon=(-\epsilon_1,\dots, -\epsilon_n)\in \{-1,1\}^n$. Thus 
$$
\tr[(A_2-A_1)\rho_{\epsilon,s}]
=\tr[(A_2+A_1)\rho_{-\epsilon,s}].
$$
The above observation implies 
\begin{align*}
&\max_{s,\epsilon}\max\left\{ 
|\tr[(A_2-A_1)\rho_{\epsilon,s}]|,
|\tr[(A_2+A_1)\rho_{\epsilon,s}]|
\right\}\\
=&\max_{s,\epsilon}\max\left\{ 
|\tr[(A_2+A_1)\rho_{-\eps,s}]|,
|\tr[(A_2+A_1)\rho_{\epsilon,s}]|
\right\}\\
=&\max_{s,\epsilon}|\tr[A\rho_{\epsilon,s}]|.
\end{align*}

All combined, we conclude that 
$$
\|A\|
\le 9\max_{s,\epsilon}\max_{s'(n+1),\epsilon'_{n+1}}|\tr[A_1\rho_{\epsilon,s}]-\delta_{s'(n+1),3}\epsilon'_{n+1}\tr[A_2\rho_{\epsilon,s}]|
=9\max_{s,\epsilon}|\tr[A\rho_{\epsilon,s}]|
$$
which finishes the proof of traceless non-homogeneous case with constant 9.

\medskip 

However, it seems that the above ``augment the number of qubits" trick does not extend to the general setting. Say, $A=A_1+A_2+A_3$ is of degree 3 and $A_k,k=1,2,3$ are its $k$-homogeneous parts. Though
$$
A'=A_1\otimes \sigma_3\otimes \sigma_3 +A_2\otimes \sigma_3\otimes \sigma_0+
A_3\otimes\sigma_0\otimes \sigma_0
$$
becomes homogeneous, it looks hopeless to repeat the same argument with constant $27=3^3$.
    

\section{Random Hamiltonians}
\label{appendix b}
	
	Let $n$ denote the number of qubits
	and $d\ll n$ be future degree of a homogeneous Hamiltonian. Recall that for any $\lS\in \{0,1,2,3\}^n$, the Pauli monomial
    $$
\sigma_{\lS}=\sigma_{\ls_1}\otimes \cdots \otimes \sigma_{\ls_n}
    $$
    has degree $d$ if $|\lS|=|\{j:\ls_j\neq 0\}|=d$.
Consider the random Hamiltonian 
    $$
    H(n, d) =\frac{1}{\sqrt{{n\choose d}}}\sum_{\lS\in \{0,1,2,3\}^n:|\lS|=d}g_{\lS}\sigma_{\lS},
    $$
    where $g_{\lS}$'s are independent standard Gaussian (or Rademacher) random variables.
	
	An interesting question is to estimate
	$$
	E(n, d):=\bE \frac1{\sqrt{n} } \|H(n, d)\|,
	$$
    for which it is common to give the estimate of this ``average maximal energy'' by comparing it with ``free energy'':
    	$$
	F(n, d, \beta):=\frac1{\beta n} \bE \log \tr e^{\beta\sqrt{n} H(n, d)}\,.
	$$
Here, let us assume $\beta>0$ for convenience (unlike the usual case where $\beta<0$). This is just for convenience and our main focus is the estimate of $E(n,d)$ anyway. 

	For example, it is easy to see that
	\begin{equation}
	\label{EF} 
	E(n, d) \le \inf_{\beta>0} F(n, d, \beta)
	\end{equation}
	and 
	\begin{equation}
	\label{FE} 
	 F(n, d, \beta) \le \frac{\log 2}{\beta} + E(n, d)
	\end{equation}
	using the simple estimate $\|A\|\le \tr(A)\le 2^n \|A\|$ for a $2^n$-by-$2^n$ positive semi-definite matrix $A$. 
	Our goal is to prove that
	\begin{equation}
	\label{F}
	F(n, d, \beta) \le \frac{\log 2}{\beta} + \beta\cdot  C\,3^d\,.
	\end{equation}
	
	Then combining   \eqref{EF} with \eqref{F} we get
	\begin{equation}
	\label{Edn}
	E(n, d) \le C \sqrt{3}^d
	\end{equation}
	by optimizing $\beta$. This is $\sqrt{\log d}$ better than in \cite{AGK}.
	
	\bigskip
	
	Our proof below is much shorter than the one in \cite{AGK}, but in fact no proof is needed as the result \eqref{Edn} follows from
	noncommutative  Khintchine inequality of Lust-Piquard \cite{LP}. An exposition with the explicit constant 
	can be found on pp. 106--107 of  Pisier's book \cite{P}. See also \cite{Ju}.
	
	\medskip

	Let us recall this inequality here. Let $\{ g_k\}_{k=1}^N$ be independent standard gaussians or Rademacher random variables. 
	Let $\{A_k\}_{k=1}^N$ be self-adjoint operators and let $\|\cdot \|_p$ be the  Schatten-$p$ norm.
	Then for $p\ge 2$, one has 
	\begin{equation}
	\label{NCKp}
	c\, \left\| \big(\sum_{k=1}^N A_k^2\big)^{1/2}\right\|_p\ 
	\le \bE\left\| \sum_{k=1}^N g_k A_k\right\|_p \le C\sqrt{p} \left\| \big(\sum_{k=1}^N A_k^2\big)^{1/2}\right\|_p\,.
	\end{equation}
		for absolute constants $c, C>0$.
	
	
	
	Denote $N=3^d\binom{n}{d}$, and write  $\binom{n}{d}^{1/2}H(n,d)=\sum_{j=1}^{N}g_k\Sigma_k$. Here, $\{ g_k\}_{k=1}^N$ are 
	 independent standard gaussians or Rademacher random variables, and $\Sigma_k^2 = \text{Id}_{2^n}$. Now \eqref{NCKp} gives us 
	$$
	\bE\|H(n,  d)\|  \le   \bE\| H(n, d)\|_{n} \le C \frac{\sqrt{n}}{{n\choose d}^{1/2}} \|\big(\sum_{k=1}^N \Sigma_k^2\big)^{1/2} \|_n
	=C \frac{\sqrt{nN}}{{n\choose d}^{1/2}}\|\text{Id}_{2^n}\|_n
	=2C\sqrt{n} \,3^{d/2}
	$$
		which is exactly \eqref{Edn}.
	
%
%
%

\bigskip



	Having this estimate from above we still want to present our proof of it that does not use noncommutative Khintchine inequality. It is just a simple ``hands-on'' proof.  It also gives some estimates on free energy in \eqref{F}. In order to prove  inequality \eqref{F}, let us first notice that the concavity of the logarithm allows us to write
	\begin{equation}
	\label{log}
	\frac1{\beta \, n}\bE \log \tr e^{\beta \sqrt{n} H(n, d)} \le \frac1{\beta \, n}\log \bE \tr e^{\beta \sqrt{n} H(n, d)}.
\end{equation}

\begin{rem} 
\label{quenched-annealed}
The left-hand side of \eqref{log} deals with the so-called quenched free energy, while the right-hand side deals with annealed free energy. (But our sign is opposite to the usually used one.)
It is easier to deal with the annealed one. This is what  \cite{AGK} does and what we treat here.
\end{rem}

\medskip
	
When one considers $ \bE \tr e^{\beta \sqrt{n} H(n, d)}$, one expands the exponential into  Taylor series. Only even powers of $H(n, d)$ contribute, because odd powers have expectation $0$.
To obtain better estimates, let us denote $K:= 2m$, $N=3^d {n \choose d}$. Then $\binom{n}{d}^{1/2}H(n,d)$ is the sum of $N$ random Pauli monomials, labled by $\binom{n}{d}^{1/2}H(n,d)=\sum_{j=1}^{N}\gamma(j)\sigma(j)$ for simplicity. Here, $\gamma(j)$'s are the $i.i.d.$. standard Gaussian random variables, and $\sigma(j)$'s are the Pauli monomials of degree $d$. 
    

    Put $\alpha=\alpha_n=\beta \sqrt{n}\binom{n}{d}^{-1/2}$ and our goal is to give a upper bound of 
 $$\tr\bE \exp\left[\alpha\sum_{1\le j\le N}\gamma(j)\sigma(j)\right]
 =\sum_{m\ge 0}\frac{\alpha^{2m}}{(2m)!}\tr \bE \left[\left(\sum_{1\le j\le N}\gamma(j)\sigma(j)\right)^{2m}\right].$$
 The odd power terms vanish, as explained earlier. Also, for each $m\ge 0$,
 \begin{align*}
     \tr \bE \left[\left(\sum_{1\le j\le N}\gamma(j)\sigma(j)\right)^{2m}\right]
     = \sum_{j_1,...j_{2m}\in [N]}\bE[\gamma(j_1)\cdots\gamma(j_{2m})]\tr[\sigma(j_1)\cdots \sigma(j_{2m})]. 
 \end{align*} 
 Since $\gamma(j)$'s i.i.d. standard Gaussian, one has $\bE \gamma(j_1)\cdots \gamma(j_{N})\neq 0$ only if $|\{k\in [N]: j_k=j\}|$ is even for all $j\in [N]$. Note that 
 $$ |\tr[\sigma(j_1)\cdots \sigma(j_{2m})] |\le 2^n,$$
 so
 \begin{equation}\label{even}
      \tr \bE \left[\left(\sum_{1\le j\le N}\gamma(j)\sigma(j)\right)^{2m}\right]
     \le 2^n\sum_{k_1+\cdots +k_N=m} \frac{(2m)!}{(2k_1)!\cdots (2k_N)!}\bE\left(|\gamma(1)|^{2k_1}\cdots |\gamma(N)|^{2k_N}\right).
 \end{equation}
 We claim that for any real numbers $a_1,\dots, a_N$, one has
  \begin{equation}\label{claim}
    \sum_{k_1+\cdots +k_N=m} \frac{(2m)!}{(2k_1)!\cdots (2k_N)!}a_1^{2k_1}\cdots a_N^{2k_N}
    \le \frac{(2m)!}{2^m m!}\left(\sum_{j=1}^{N}a_j^2\right)^m.
 \end{equation}
 We will verify this claim below. Then \eqref{claim}, together with \eqref{even}, gives
  \begin{equation*}
     \tr \bE \left[\left(\sum_{1\le j\le N}\gamma(j)\sigma(j)\right)^{2m}\right]
     \le 2^n\frac{(2m)!}{2^m m!}\bE \left[\left(\sum_{1\le j\le N}|\gamma(j)|^2\right)^{m}\right].
  \end{equation*} 
 Therefore, we have the upper bound
 \begin{align*}
    \tr\bE \exp\left[\alpha\sum_{1\le j\le N}\gamma(j)\sigma(j)\right]
    &\le 2^n\sum_{m\ge 0}\frac{\alpha^{2m}}{(2m)!}\cdot \frac{(2m)!}{2^m m!}\bE \left[\left(\sum_{1\le j\le N}|\gamma(j)|^2\right)^{m}\right]\\
    &=2^{n}\bE e^{\frac{\alpha^2}{2} \sum_{j=1}^{N}|\gamma(j)|^2}.
 \end{align*}
Since $\gamma(j)$'s are i.i.d. standard Gaussian, 
$$
\bE e^{c\sum_{j=1}^{N}|\gamma(j)|^2}
=(\bE e^{c|\gamma(1)|^2})^N=(1-2c)^{-N/2},\qquad 0\le c<1/2.
$$
Thus, for $\beta$ such that 
$$
\alpha^2=\beta^2 n \binom{n}{d}^{-1}<1,
$$
we obtain the estimate 
\begin{equation}
    \tr\bE \exp\left[\alpha\sum_{1\le j\le N}\gamma(j)\sigma(j)\right]
    \le 2^n(1-\alpha^2)^{-N/2}.
\end{equation}
Using the elementary inequality
$$
\log(1-x)\ge -2x, \qquad 0<x<1/2,
$$
we have 
\begin{equation}\label{final estimate}
    \frac1{\beta n}\log \tr \bE e^{\beta\sqrt{n} H(n, d)}
\le \frac{\log 2}{\beta}-\frac{N}{2\beta n}\log(1-\alpha^2)
\le \frac{\log 2}{\beta}-\frac{N}{2\beta n}\cdot (-2\alpha^2)
=\frac{\log 2}{\beta}+3^d\beta
\end{equation}
 for all $\beta$ such that
\begin{equation}\label{beta constraint}
    \alpha^2=\beta^2 n \binom{n}{d}^{-1}<\frac{1}{2}.
\end{equation}

To conclude, we have shown that given claim \eqref{claim}, for all $\beta$ in \eqref{beta constraint}
$$
E(n,d)\le F(n,d,\beta)=\frac1{\beta n}\log \tr \bE e^{\beta\sqrt{n} H(n, d)}
\le \frac{\log 2}{\beta}+3^d\beta.
$$
This gives 
$$E(n,d)\le 2\sqrt{\log 2} \cdot \sqrt{3}^d
$$
by choosing $\beta$ such that 
$$
\beta^2=3^{-d}\log 2.
$$
This choice is not against the constraint \eqref{beta constraint}, since $3^{-d}\log 2<\frac{1}{2n}\binom{n}{d}$ is satisfied whenever $n\ge d\ge 1$.

Now, it remains to prove the claim \eqref{claim}. It is equivalent to 
\begin{equation}
     \sum_{k_1+\cdots +k_N=m} \frac{(2m)!}{(2k_1)!\cdots (2k_N)!}a_1^{2k_1}\cdots a_N^{2k_N}
    \le \frac{(2m)!}{2^m m!}  \sum_{k_1+\cdots +k_N=m} \frac{m!}{k_1!\cdots k_N!}a_1^{2k_1}\cdots a_N^{2k_N}.
\end{equation}
So it suffices to compare the coefficients before each monomial
$$
\frac{(2m)!}{(2k_1)!\cdots (2k_N)!}
\le \frac{(2m)!}{2^m m!} \cdot \frac{m!}{k_1!\cdots k_N!}, \qquad \forall k_1+\cdots +k_N=m
$$
which is nothing but 
$$
(2k_1)!\cdots (2k_N)!\ge 2^m k_1!\cdots k_N!, \qquad \forall k_1+\cdots +k_N=m.
$$
To see this, note that 
$$
\frac{(2k)!}{k!}=(2k)(2k-1)\cdots (k+1)\ge (2k)(2k-2)\cdots (2)=2^k k!.
$$
This implies, recalling the constraint $k_1+\cdots +k_N=m$,
$$
(2k_1)!\cdots (2k_N)!
\ge 2^{k_1}k_1 !\cdots 2^{k_N}
k_N !
=2^{k_1+\cdots +k_N} k_1!\cdots k_N!
=2^m k_1!\cdots k_N!.
$$
This completes the proof of the claim and thus the desired bound $C\sqrt{3}^d$ for $E(n,d)$.

\section{Possible extension to the qudit system}
\label{qudit}
It is possible to extend our main results on qubit systems to qudit systems. We only highlight the main ingredients here, and for statements about the qudit systems without proofs, we refer to \cite{SVZ} for details. 

Let $\go\ge 3$ be a prime integer and denote $\omega=\omega_\go=e^{2\pi i/\go}$. Let $\Z_\go$ and $\Omega_\go$ be the additive and multiplicative groups of order $\go$, respectively. The Heisenberg-Weyl basis of $M_\go(\C)^{\otimes n}$ is the class of matrices
$$
X^{\px}Z^{\pz},\qquad (\px,\pz)\in \Z_\go\times \Z_\go
$$
where $X$ and $Z$ are the shift and clock matrices, respectively
$$
X\ket{j}=\ket{j+1},\qquad Z\ket{j}=\omega \ket{j},\qquad j\in \Z_\go.
$$
The Heisenberg-Weyl decomposition of any $A\in M_\go(\C)^{\otimes n}$ is 
\begin{equation}
    A=\sum_{(\vpx, \vpz)\in \Z_\go^n\times \Z_\go^n}\widehat{A}(\vpx,\vpz)X^{\vpx}Z^{\vpz},\qquad \widehat{A}(\vpx,\vpz)\in \C, \qquad X^{\vpx}Z^{\vpz}:=\otimes_{j\in [n]}X^{\px_j}Z^{\pz_j}
\end{equation}
We define the degree of $A$ as 
\begin{equation}
    \deg_0(A):=\max_{\widehat{A}(\vpx,\vpz)\neq 0}|(\vpx,\vpz)|
\end{equation}
where we put 
\begin{equation}
|(\vpx,\vpz)|:=|\{j\in [n]: (\px_j,\pz_j)\neq (0,0)\}|.
\end{equation}
Note that there are alternative definitions of degree, such as 
\begin{equation}
    \deg(A):=\max_{\widehat{A}(\vpx,\vpz)\neq 0}\sum_{j\in [n]} \px_j +\pz_j,\qquad 0\le \px_j, \pz_j\le \go-1. 
\end{equation}
We will see why we used $\deg_0(A)$ here, but they are comparable 
\begin{equation}
    \deg_0(A)\le \deg(A)\le 2(\go-1)\deg_0(A)
\end{equation}
up to a factor independent of $n$. So, the choice of degree here does not affect much in describing the locality of $A$.

Since $\go$ is prime, we may decompose the group $\Z_\go\times \Z_\go$ as 
\begin{equation}\label{decomposition}
    \Z_\go\times \Z_\go
    =\bigcup_{(\gx,\gz)\in \Sigma}\langle (\gx, \gz)\rangle
\end{equation}
Here, for an element $g$ of a group $G$ we used the convention that $\langle g\rangle$ denotes the subgroup of $G$ generated by $g$. The set $\Sigma$ of generators is given by 
\begin{equation}
    \Sigma=\{(0,1),(1,1),(2,1),\dots, (\go-1,1),(1,0)\}.
\end{equation}
Note that, $|\Sigma|=\go+1$, and the intersection of each of two subgroups in the decomposition \eqref{decomposition} is exactly the singleton $\{(0,0)\}$ of the unit element. Moreover, for any $(\gx,\gz)\in \Sigma$, the set of eigenvalues of $X^{\gx}Z^{\gz}$ is $\Omega_\go$, each having multiplicity exactly one. For any $(\gx,\gz)\in \Sigma$ and $z\in \Omega_\go$, we write $\ket{e^{\gx,\gz}_{z}}$ as the unit eigenvector of $X^{\gx}Z^{\gz}$ with eigenvalue $z$.

For any $(\gx,\gz)\in \Sigma$, $X^{\gx} Z^{\gz}$ generates a commutative subalgebra of $M_{\go}(\C)$ that is exactly 
\begin{equation}
    \ca_{\gx,\gz}:=\spa \{X^{k\gx}Z^{k\gz}:k\in \Z_\go\}.
\end{equation}
Here, we used the fact that $(X^{\gx}Z^{\gz})^k=\omega^{\frac{1}{2}k(k-1)\gx\gz}X^{k\gx}Z^{k\gz}$. Let $\ce_{\gx,\gz}$ be the conditional expectation from $M_\go(\C)$ onto $\ca_{\gx,\gz}$. Then it has the form 
\begin{equation}
   \ce_{\gx,\gz}(A)
   =\sum_{z\in \Omega_\go}\ket{e^{\gx,\gz}_{z}}\bra{e^{\gx,\gz}_{z}}A\ket{e^{\gx,\gz}_{z}}\bra{e^{\gx,\gz}_{z}},\qquad A\in M_\go(\C).
\end{equation}

Now, for any $(\vgx,\vgz)=\{(\gx_j,\gz_j)\}_{j\in [n]}\in \Sigma^n$, we denote by $\ce_{\vgx,\vgz}$ the conditional expectation from $M_\go(\C)^{\otimes n}$
onto the commutative subalgebra 
\begin{equation}
    \ca_{\vgx,\vgz}:=\spa\{\ca_{\gx,\gz}\otimes \un \otimes \cdots\otimes \un, \dots, \un\otimes \cdots \otimes \un\otimes \ca_{\gx,\gz}\}.
\end{equation}
It takes the explicit form
\begin{equation}
   \ce_{\vgx,\vgz}(A)
   =\sum_{\z\in \Omega_\go^n}\ket{e^{\vgx,\vgz}_{\z}}\bra{e^{\vgx,\vgz}_{\z}}A\ket{e^{\vgx,\vgz}_{\z}}\bra{e^{\vgx,\vgz}_{\z}},\qquad \ket{e^{\vgx,\vgz}_{\z}}=\otimes_{j\in [n]}\ket{e^{\gx_j,\gz_j}_{z_j}},\qquad A\in M_\go(\C)^{\otimes n}.
\end{equation}

\begin{lemma}
    For any $A\in M_\go(\C)^{\otimes n}$, we have 
    \begin{equation}
        \frac{1}{(\go+1)^n}\sum_{(\vgx,\vgz)\in \Sigma^n}\ce_{\vgx,\vgz}(A)
        =\sum_{(\vpx, \vpz)\in \Z_\go^n\times \Z_\go^n}
        (\go+1)^{-|(\vpx, \vpz)|}\widehat{A}(\vpx,\vpz)X^{\vpx}Z^{\vpz}.
    \end{equation}
\end{lemma}

\begin{proof}
    The proof is similar to the qubit case. By definition, we have for all $(\gx,\gz)\in \Sigma$ and $(\px,\pz)\in \Z_\go\times \Z_\go$ that 
    \begin{equation}
    \ce_{\gx,\gz}(X^{\px}Z^{\pz})
    =\begin{cases}
       X^{\px}Z^{\pz} & (\px,\pz)\in \langle (\gx, \gz)\rangle  \\
       0 & \textnormal{otherwise}
    \end{cases}.
    \end{equation}
This implies immediately that for all $(\vgx,\vgz)\in \Sigma^n$ and $(\vpx,\vpz)\in \Z_\go^n\times \Z_\go^n$ 
\begin{equation}
    \ce_{\vgx,\vgz}(X^{\vpx}Z^{\vpz})
    =\begin{cases}
       X^{\vpx}Z^{\vpz} & (\px_j,\pz_j)\in \langle (\gx_j, \gz_j)\rangle \text{ for all }j\in [n] \\
       0 & \textnormal{otherwise}
    \end{cases}.
    \end{equation}
Then, using the fact that each $(\px,\pz)\neq(0,0)$ belongs to exactly one of $\langle (\gx,\gz)\rangle,(\gx,\gz)\in \Sigma$, we have 
\begin{equation}
    \sum_{(\vgx,\vgz)\in \Sigma^n} \ce_{\vgx,\vgz}(X^{\vpx}Z^{\vpz})
    = \sum_{j, (\gx_j, \gz_j):(\px_j,\pz_j)\in \langle (\gx_j, \gz_j)\rangle }X^{\vpx}Z^{\vpz}
    =(\go+1)^{n-|(\vpx,\vpz)|}X^{\vpx}Z^{\vpz}.
\end{equation}
This finishes the proof of the desired equality by linearity.
\end{proof}

Similar to the qubit case, the above lemma also helps in improving the constant of the reduction method for BH inequality on qudit systems. We omit the details here, since the BH constant on cyclic groups are not good enough.

\bigskip

To treat the non-homogeneous case, we also need a Figiel's inequality in this case. The main ingredient is the contractivity of the linear map $P_r$ defined by
\begin{equation}
    P_r: X^{\vpx}Z^{\vpz}\mapsto r^{|(\vpx,\vpz)|}X^{\vpx}Z^{\vpz}
\end{equation}
when $r\in [-1,1]$. The contraction property of $P_r$ is trivial when $r\in [0,1]$, since it is again the tensor product of the depolarizing channel 
$$
P_r (A)=(r\, A +(1-r)\go^{-1}\tr[A ]\un)^{\otimes n}.
$$
It remains to prove the property when $r=-1$, since $P_{-r}=P_{-1}P_r$.

However, the contraction property fails for $r=-1$ when $\go\ge 3$ even when $n=1$. When $n=1$, our map $P_{-1}$ is given by 
$$
P_{-1}(A)=-A+2K^{-1}\tr(A)\un.
$$
There is a naive estimate 
$$
\|P_{-1}(A)\|
=\|-A+2K^{-1}\tr(A)\un\|
\le 3\|A\|,
$$
and in general $P_{-1}$ is not a contraction. Indeed, take $\go=3$ and let $A$ be the diagonal matrix with diagonal entries $1,1,-1$. Then $P_{-1}(A)$ is the diagonal matrix with diagonal entries 
$-1/3,-1/3,5/3$. So it cannot be a contraction.

In other words, in the high-dimensional setting, we cannot expect 
$$
\|P_{-1}(A)\|\le C\|A\|
$$
with $C$ independent of $n$ for any $A\in M_\go(\C)^{\otimes n}$. But we only need it to be true for low-degree $A$. For this, one can use the estimate in \cite{BKSVZ}, following the arguments in Section \ref{card}.


\section{Large deviation of $\|H(d, n)\|_{op}$}
\label{explan}

Recall that $N=3^d{n \choose d}$.
Let 
\begin{equation}
\label{Delta}
\Delta := \sup_{\rho\, \, \text{state}} \frac1{N} \sum_{i=1}^N \big(\tr [\Sigma_i \rho]\big)^2
\end{equation}
In  \cite{BBoVH} of A. S. Bandeira, M. T. Boedihardjo, and R. van Handel  it was proved that
\begin{equation}
\label{March2}
\bP\{ \|{\bf H(d, n)}\|_{op} -\bE \|{\bf H (d, n)}\|_{op} >t\} \le e^{-\frac{t^2}{2\Delta}}\,,
\end{equation}
where, by reconciling notations, we have
$$
{\bf H(d, n)} = \frac1{\sqrt{3^d { n\choose d}} }\sum_{i=1}^N g_i\Sigma_i = \frac1{\sqrt{3^d}} H(d, n)\,.
$$
Therefore,
$$
\bP\{ \|H(d, n)\|_{op} -\bE \|H (d, n)\|_{op} >3^{d/2}t\} \le e^{-\frac{t^2}{2\Delta}}\,.
$$
which is
$$
\bP \{ \|H(d, n)\|_{op} -\bE \|H (d, n)\|_{op} > t\} \le e^{-\frac{t^2}{2\cdot 3^d\Delta}}\,.
$$
Also Proposition B.1 of \cite{ACKK} (see also Lemma \ref{Delta} below, where we repeat \cite{ACKK}) proves that

$$
3^d \le 2n+1 \Rightarrow \Delta \le 3^{-d}\,.
$$
Therefore, in this regime
\begin{equation}
\label{March3}
\bP\{ \|H(d, n)\|_{op} -\bE \|H (d, n)\|_{op} > t\} \le e^{-\frac{t^2}{2}}\,,
\end{equation}
which, in its turn, implies

$$
\bE e^{\beta \,\| H(d, n)\|_{op}- \beta \, \bE\|H(d, n)\|_{op}}   \le
e^{	4\beta^2 }.
$$
So,
\begin{equation}
\label{eE}
e^{\beta\, \sqrt{n} \, \bE\|H(d, n)\|_{op}}  \ge \bE e^{\beta\, \sqrt{n} \, \|H(d, n)\|_{op}} \cdot e^{-4 n\beta^2}\,.
\end{equation}

In fact, for any $d, n$ one has
\begin{equation}
\label{delta23}
\Delta \le \Big(\frac23\Big)^d,
\end{equation}
which gives
\begin{equation}
\label{March2d}
\bP\{ \|H(d, n)\|_{op} -\bE \|H (d, n)\|_{op} > t\} \le e^{-\frac{t^2}{2\cdot 2^d}}\,,
\end{equation}
Inequality \eqref{delta23} is also proved in \cite{ACKK}, but we give a somewhat different proof below in Lemma \ref{Delta23}.

\bigskip

\subsection{Some consequences of large deviation estimates}
Given a positive r. v. $g$ such that $\bE g=1$ and
\begin{equation}
\label{g}
\bP\{ g-1 > s\} \le  e^{-\frac{s^2 a^2}{2}}
\end{equation}
one can easily see that
\begin{align*}
& \bE g^{2m} \le \bE (\1_{g \le 1} g^{2m} )+ 1+
2m \int_1^\infty s^{2m-1} e^{-\frac{(s-1)^2 a^2}{2}} \, ds\le 
\\
& 2+ 2m \int_0^\infty (1+s)^{2m-1} e^{-\frac{s^2 a^2}{2}} \, ds= 2+ \frac{2m}{a^{2m}} \int_0^\infty (a+s)^{2m-1} e^{-\frac{s^2 }{2}} \, ds \le
\\
&2+  \frac{2m}{a^{2m}}\big(\int_0^a \dots + \int_a^\infty \dots\big) \le  2 +  \frac{2m}{a^{2m}}\big(2^{2m-1} a^{2m-1} + 2^m \int_a^\infty s^{2m-1} e^{-s^2/2} \, ds\big)
\\
&2 +  \frac{2^{2m}m}{a} + \frac{2m}{a^{2m}}2^m (m-1)! 2^{m-1}= 2 +  \frac{2^{2m} m}{a} + \frac{1}{a^{2m}}2^{2m} m!\,.
\end{align*}

Putting $f:=\|H(d, n)\|_{op}$, $g:= \frac{f}{\bE f}$, $ a=\bE f$ we get \eqref{g} from \eqref{March3}.

\medskip

Then multiplying the previous inequality by $a^{2m}=(Ef)^{2m}$ we get
\begin{align}
\label{f-m}
&\bE [\|H(d, n)\|_{op}^{2m}]\le 2\big( \bE \|H(d, n)\|_{op}\big)^{2m} + \notag
\\& 2^{2m} m \big( \bE \|H(d, n)\|_{op}\big)^{2m-1} + 2^{2m}m!\,.
\end{align}

By inequality \eqref{asymptH} proven below, the left hand side has an estimate
\begin{equation}
\label{m2n}
m \le \frac{n}{2C d^2}\Rightarrow \bE \big[\|H(d, n)\|_{op}^{2m}\big] \ge c^m \, 3^{dm} \, m^m\,.
 \end{equation}

So,
$$
2  \big( \bE \|H(d, n)\|_{op}\big)^{2m} +2^{2m} m \big( \bE \|H(d, n)\|_{op}\big)^{2m-1} \ge 
$$
$$
c^m \, 3^{dm} \, m^m - 2^{2m} m!\ge \frac12 \,  c^m \, 3^{dm} \, m^m,
 $$
for $d\ge 2$. This of course means that $\bE \|H(d, n)\|_{op}\ge 1$, and so,
 $$
 \big( \bE \|H(d, n)\|_{op}\big)^{2m-1} \le  \big( \bE \|H(d, n)\|_{op}\big)^{2m}.
 $$
 Therefore,
 \begin{equation}
 \label{cm3dm}
 (2+ 2^{2m} m)\big( \bE \|H(d, n)\|_{op}\big)^{2m} \ge \frac12 c^m 3^{dm} m^m\,.
 \end{equation}

Hence,

\begin{equation}
\label{final}
\bE \|H(d, n)\|_{op} \ge \tilde c\, 3^{d/2} \sqrt{m}, \quad m\le \frac{n}{C\,d^2}\,.
\end{equation}

The fact that we used $3^d \le 2n+1$ is not important here, the Proposition B.2 of \cite{ACKK} (see also Lemma \ref{Delta23} below) proves the same \eqref{March3} without any restrictions  on $d, n$ in the form
\begin{equation}
\label{March4}
\bP\{ \|H(d, n)\|_{op} -\bE \|H (d, n)\|_{op} > t\} \le e^{-\frac{t^2}{2\cdot 2^d}}\,,
\end{equation}

This changes \eqref{f-m} but only slightly and it becomes
\begin{equation}
\label{f-m=d}
\bE [\|H(d, n)\|_{op}^{2m}]\le \big( \bE \|H(d, n)\|_{op}\big)^{2m} + 2^{dm}\,2^m m!\,.
\end{equation}
So,
$$
 \big( \bE \|H(d, n)\|_{op}\big)^{2m} \ge c^m \, 3^{dm} \, m^m - 2^{dm}\,2^{2m} m!\ge \frac12 \,  c^m \, 3^{dm} \, m^m,
 $$
for $d\ge 4$.

\medskip

Therefore, we again have \eqref{final}.



\medskip

\begin{rem}
\label{comm-case}
The case  $ d^2\le c n$ is a ``{\it commutative case}".  This property says that the probability  to commute for
$\Sigma_a, \Sigma_b$ drawn uniformly at random and independently is $\ge 1-c$.
\end{rem}

\section{An estimate from below $\frac{1}{\sqrt{n}}\bE \|H(d, n)\|_{op} \ge c\, \frac{3^{d/2}}{d}$. Proof of  \eqref{asymptH}}
\label{another}

As always $N= 3^d {n \choose d}$, it is convenient to slightly change the normalization and to consider
$
\bH(d, n) = \frac1{\sqrt{N}}\sum_{j=1}^N g_j \Sigma_j,
$
where $g_j$ are independent standard gaussians, and $\Sigma_j$ are all tensor product of $n$ Paulis $\sigma_0, \sigma_1, \sigma_2, \sigma_3$ such that among them only $d$
are not $\sigma_0= \text{Id}_{2\times 2}$.

Let us denote 
$$
b_m :=\bE\tr [\bH(d, n)^{2m}],
$$
we will have a recursive inequality on $b_{m+1}$.

\medskip

We will be using gaussian integration by parts, which is the following formula, where $F$ is a smooth functions and $g, g_1, \dots, g_\ell$ are any gaussians:
\begin{equation}
\label{g-by-parts}
\bE ( g F(g_1, \dots, g_n) = \sum_{k=1}^n \bE(g g_\ell) \bE \big( F'_\ell (g_1, \dots, g_n)\big)\,.
\end{equation}
Also
\begin{equation}
\label{diff}
\big( \bH(g_1,\dots, g_N)^{2m+1}\big)'_j = \sum_{k=0}^{2m} \bH^k \, \big( \bH(g_1,\dots, g_N)\big)'_j \, \bH^{2m-k}\,.
\end{equation}

Therefore,
\begin{align*}
& b_{m+1} :=\bE\tr [\bH(d, n)^{2m+2}] =\frac1{N^{1/2}}\sum_{j=1}^N \bE\tr \Big(\Sigma_j g_j \bH(d, n)^{2m+1}\Big) =^{\eqref{g-by-parts}, \eqref{diff}}
\\
&\frac1{N}\sum_{j=1}^N\sum_{k=0}^{2m} \bE\tr \big( \Sigma_j \,\bH^k \,\Sigma_j \,\bH ^{2m-k}\big) = \frac1{N}\sum_{j=1}^N\sum_{k=0}^{2m} \bE\tr \big( \Sigma_j ^2 \,\bH^k \,\,\bH ^{2m-k}\big)  + 
\\
&\frac1{N}\sum_{j=1}^N\sum_{k=0}^{2m} \bE\tr \big( \Sigma_j \,[\bH^k, \Sigma_j ]\,\bH ^{2m-k}\big)  = \frac1{N}\sum_{j=1}^N\sum_{k=0}^{2m} \bE\tr \big( \text{Id}_{2^n\times 2^n} \,\bH ^{2m}\big)  + 
\\
& \frac1{N}\sum_{j=1}^N\sum_{k=0}^{2m} \sum_{r=0}^{k-1} \bE\tr \big( \Sigma_j \, \bH^r\, [\bH, \Sigma_j] \,\bH^{k-r-1} \, \bH^{2m-k}\big) = \frac1{N}\sum_{j=1}^N\sum_{k=0}^{2m} \bE\tr \big( \bH ^{2m}\big) +
\\
& \frac1{N^{3/2}} \sum_{j=1}^N\sum_{j'=1}^N\sum_{k=0}^{2m} \sum_{r=0}^{k-1} \bE\tr \big( \Sigma_j \, \bH^r\, g_{j'}\,[\Sigma_{j'}, \Sigma_j]  \, \bH^{2m-r-1}\big) = 2m\, b_m +
\\
&\frac1{N^2} \sum_{j, j'=1}^N \sum_{k=0}^{2m} \sum_{r=0}^{k-1} \sum_{\ell=0}^{r-1} \bE \tr \big(\Sigma_j \bH^\ell\, \Sigma_{j'}\, \bH^{r-\ell-1} \, [\Sigma_{j'}, \Sigma_j]\, \bH^{2m-r-1}\big) +
\\
& \frac1{N^2} \sum_{j, j'=1}^N \sum_{k=0}^{2m} \sum_{r=0}^{k-1} \sum_{\ell=0}^{2m-r-2} \bE \tr \big(\Sigma_j \bH^{r} \, [\Sigma_{j'}, \Sigma_j]\, \bH^\ell \Sigma_{j'}\bH^{2m-r-\ell-2}\big)  =: 2m b_m + E_1+ E_2.
\end{align*}

We would like to record also one equality in this chain:
\begin{equation}
\label{SiSi}
b_{m+1}= \frac1{N}\sum_{j=1}^N\sum_{k=0}^{2m} \bE\tr \big( \Sigma_j \,\bH^k \,\Sigma_j \,\bH ^{2m-k}\big) \,.
\end{equation}

Let us estimate $E_1$ (the same will be true for $E_2$): $\cH :=\bH^{2m-2}$. By using H\"older inequality  for trace
$$
|\tr \big(\Sigma_j \bH^\ell\, \Sigma_{j'}\, \bH^{r-\ell-1} \, [\Sigma_{j'}, \Sigma_j]\, \bH^{2m-r-1}\big) | \le
$$
$$
 \|\Sigma_j\|_\infty \| \bH^{\ell}\|_a \|\Sigma_{j'}\|_\infty  \| \bH^{r-\ell-1}\|_{b}  \|[\Sigma_{j'}, \Sigma_j]\|_\infty \| \bH^{2m-r-1}\|_{c}
$$
where $\|\cdot\|_p$ means Schatten--von Neumann norm of the matrix $\cdot$ and  $a=\frac{2m-2}{\ell}$, $b=\frac{2m-2}{r-\ell-1}$, $c=\frac{2m-2}{2m-r-1}$, and
$$
\frac1a  + \frac 1b +\frac 1c =1, a, b, c \ge 1\,.
$$
Thus each term of $E_1$ has the following estimate
\begin{equation}
\label{term}
\!\!\!\!\!\!\!|\tr \!\big(\Sigma_j \bH^\ell\, \Sigma_{j'}\, \bH^{r-\ell-1} \, [\Sigma_{j'}, \Sigma_j]\, \bH^{2m-r-1}\big) \!| \le \|[\Sigma_{j'}, \Sigma_j]\|_\infty \big(\tr \bH^{2m-2}\big)\,.
\end{equation}

Numbers $ r-\ell-1, 2m-r-1, \ell$ sum up to $2m-2$, so  $E_1$ can be rewritten (using non-negative integers $\alpha, \beta, \gamma$) as
$$
E_1= \frac1{N^2} \sum_{j, j'=1}^N \sum_{\alpha, \beta, \gamma: \alpha+\beta+\gamma= 2m-2} \bE \tr \big(\Sigma_j \bH^\alpha\, \Sigma_{j'}\, \bH^{\beta} \, [\Sigma_{j'}, \Sigma_j]\, \bH^{\gamma}\big) 
$$

So number of terms in $E_1$ is $N^2\times \frac{(2m-2)(2m-2-1)}{2}=CN^2m^2$. Taking into account \eqref{term}  we get
\begin{equation}
\label{E1}
E_1 \le Cm^2\, \bE \tr \bH^{2m-2}\ \frac1{N^2} \sum_{j', j=1}^N \|[\Sigma_{j'}, \Sigma_j]\|_\infty \le \frac{Cm^2 d^2}{n}\, \bE \tr \bH^{2m-2}\,.
\end{equation}
Verbatim the same reasoning shows that 
\begin{equation}
\label{E2}
E_2 \le Cm^2\, \bE \tr \bH^{2m-2}\ \frac1{N^2} \sum_{j', j=1}^N \|[\Sigma_{j'}, \Sigma_j]\|_\infty \le \frac{Cm^2 d^2}{n}\,\bE \tr \bH^{2m-2}\,.
\end{equation}

Coming back to
$$
b_{m+1} = 2m\, b_m + E_1+ E_2
$$
we obtain 
\begin{equation}
\label{recurs}
b_{m+1} \ge 2m\, b_m - Cm^2 \eps b_{m-1},\quad \eps:= \frac{d^2}{n} ,.
\end{equation}
We rewrite it as follows
$$
\frac{b_{m+1}}{b_m} \ge 2m - C\eps m^2 \frac{b_{m-1}}{b_m}=  m\Big( 2- C\eps\, m \frac{b_{m-1}}{b_m}\Big)\,.
$$
We denote $y_m :=  \frac{b_{m+1}}{m\,b_m}$ to get
\begin{equation}
\label{ym}
y_m\ge  2- C\eps \frac1{y_{m-1}}, \quad m\ge 2\,.
\end{equation}

Notice that $b_2 =\bE \tr \bH^4$ and $b_1 =\bE \tr \bH^2 $, have the relationship
$$
b_2 \ge  b_1\,.
$$
So, $1/y_1\le 1$. Then, just by \eqref{ym} all $1/y_k\le 1$ if $\eps $ is small enough, in particular $b_{k-1} \le b_k$.
This automatically means that  
\begin{equation}
\label{bm}
b_{m+1} \ge   m\Big( 2-  C\eps\, m \Big) \, b_m\,.
\end{equation}

\medskip

One can use this inequality for estimating $b_{m+1}$ from below as many times as 
\begin{equation}
\label{mnd2}
m\le \frac{1}{C} \cdot \frac1\eps =\frac{1}{C} \cdot \frac{ n}{d^2}
\end{equation}
getting, as a result,
$$
\bE \tr \bH^{2m+2} \ge (2-C_1\eps)\, m \, \bE \tr \bH^{2m}\,,\quad m \le \frac{c\, n}{d^2}\,
$$
We got (just by using $\tr \bH^2 =\tr \text{Id}_{2^n\times 2^n}= 2^n$)
\begin{equation}
\label{nbfH}
2^n \bE \| \bH^{2m+2}\|_{op} \ge \bE \tr \bH ^{2m+2} \ge 2^n\, \Big(2-C_1\frac{d^2}{n}\Big)^m \, m!\,.
\end{equation}

Coming back to our notation $H(d, n)$ which is $3^{d/2} \bH(d, n)$ we obtain
\begin{equation}
\label{asymptH}
 2^n\,\bE \| H(d, n)^{2m}\|_{op} \ge \bE \tr H(d, n) ^{2m} \ge 2^n\,c^m \, 3^{dm} m^m\,.
\end{equation}

Now we use \eqref{g} to get again \eqref{f-m}, which is
\begin{align}
\label{f-m1}
&\bE [\|H(d, n)\|_{op}^{2m}]\le 2\big( \bE \|H(d, n)\|_{op}\big)^{2m} + \notag
\\
& 2^{2m} m \big( \bE \|H(d, n)\|_{op}\big)^{2m-1} + 2^{2m}m!\,.
\end{align}
This of course gives
$$
(2+ 2^{2m} m)  \big( \bE \|H(d, n)\|_{op}\big)^{2m} \ge \bE \big( \|H(d, n)\|_{op}^{2m}\big) - 2^m m!\,.
$$
We now use \eqref{asymptH} to get 
$$
(2+ 2^{2m} m)  \big( \bE \|H(d, n)\|_{op}\big)^{2m} \ge c^m 3^{dm} m^m - 2^m m! \ge^{d \ge d_0} \frac12 c^m 3^{dm} m^m \,.
 $$
 we mentioned this already in \eqref{m2n}.
 Now we can choose $m = \frac{c_0 n}{d^2}$.
 
\medskip

So put $m = \frac{c_0 n}{d^2}$ and we obtain
\begin{equation}
\label{Final}
\bE \|H(d, n)\|_{op} \ge c\, 3^{d/2} \sqrt{\frac{n}{d^2}} = c \sqrt{n} \frac{3^{d/2}}{d}\,.
\end{equation}
 
 Thus, this expectation $\bE \|H(d, n)\|_{op}$ happily scales as $\sqrt{n}$ from below as it  was doing from above.
 
 Estimate \eqref{Final} was obtained under the assumption that $\eps = \frac{d^2}{n} \le \eps_0$ a certain small absolute positive number.
 
 Of course our proof here reminds a lot the proof in \cite{ACKK}, while formally it is a different  and arguable a little bit more direct proof.
 
\subsection{One more estimate from above}
\label{above again}
Recall that  we do not need to prove the estimate from above
$$
\bE \|H(d, n)\|_{op} \le c\, 3^{d/2} \sqrt{n}
$$
as it follows directly from Francoise Lust-Piquard NCK noncommutative Khintchine inequality \cite{LP} (see also p. 106 in Gilles Pisier's book \cite{P}).

\medskip

But the following independent proof might become handy for sharper estimates under certain $d, n$ regimes.
So, again, by \eqref{SiSi}
 $$
 b_{m+1} =\bE\tr [\bH(d, n)^{2m+2}] 
  =\frac1{N}\sum_{j=1}^N\sum_{k=0}^{2m} \bE\tr \big( \Sigma_j \,\bH^k \,\Sigma_j \,\bH ^{2m-k}\big).
$$
Denoting
$$
\phi(A):=\frac{1}{N}\sum_{j=1}^{N} \Sigma_j A \Sigma_j
$$
we may write 
 $$
 b_{m+1} =\sum_{k=0}^{2m} \bE\tr \big[ \phi(\bH^k)\bH ^{2m-k}\big].
$$
By definition, $\phi$ is a unital completely positive trace preserving map, and it is known that it is a contraction over Schatten-$p$ classes for all $p\ge 1$ \cite{B}: for all $A$,
$$
\|\phi(A)\|_p\le \|A\|_p,\qquad  p\ge 1.
$$
This, together with H\"older's inequality, implies for all $0\le k\le 2m$ that
\begin{align*}
 \bE\tr \big[ \phi(\bH^k)\bH ^{2m-k}\big]
& \le \left[\bE\tr\left(\phi(\bH^k)^{2m/k}\right)\right]^{k/2m}\left[\bE\tr\left(\bH^{2m}\right)\right]^{(2m-k)/2m}\\
& \le  \left[\bE\tr\left(\bH^{2m}\right)\right]^{k/2m}\left[\bE\tr\left(\bH^{2m}\right)\right]^{(2m-k)/2m}\\
&=\bE\tr\left(\bH^{2m}\right).
\end{align*} 
Then, we just proved 
$$
 b_{m+1} =\sum_{k=0}^{2m} \bE\tr \big[ \phi(\bH^k)\bH ^{2m-k}\big]\le (2m+1)\bE\tr\left(\bH^{2m}\right)
 =(2m+1)b_{m}.
$$
Inductively, one obtains
$$
b_m\le (2m-1)b_{m-1}\le \cdots \le (2m-1)!!b_1
$$
where 
$$
b_1=\bE\tr\left(\bH^2\right)=\frac{1}{N}\sum_{i=1}^{N}\bE[g_i^2]\tr\text{id}_{2^n}=2^n.
$$
Therefore, we have shown that 
$$
b_m\le 2^n\cdot (2m-1)!!=2^n\cdot \frac{(2m)!}{2^m m!}.
$$
In our earlier work, we obtained the estimate 
$$
b_m\le 2^n\cdot \frac{(2m)!}{2^m m!}\frac{1}{N^m}\bE\left[\left(\sum_{1\le j\le N}|g_j|^2 \right)^m\right]\,.
$$
With the new bound, we estimate 
\begin{align*}
\bE\tr [e^{\beta\sqrt{n} H(n,d)}]
&=\bE\tr [e^{\sqrt{3}^d\beta\sqrt{n} \bH}]\\
&=\sum_{m\ge 0}\frac{3^{dm}\beta^{2m}n^m}{(2m)!}b_m\\
&\le 2^n\sum_{m\ge 0}\frac{3^{dm}\beta^{2m}n^m}{(2m)!}\frac{(2m)!}{2^m m!}\\
&=2^n e^{\frac{3^d n\beta^2}{2}}
\end{align*}
and thus
$$
\frac{1}{\beta n}\log \bE\tr [e^{\beta\sqrt{n} H(n,d)}]\le \frac{\log 2}{\beta}+\frac{3^d\beta}{2}
$$
Optimizing over $\beta>0$ yields 
$$
E(d,n)\le \sqrt{2\log 2}\cdot \sqrt{3}^d.
$$

\bigskip
 
 Now one needs to use \cite{BBoVH}.

\section{Using \cite{BBoVH} for the second term of the asymptotics}
 \label{2term}
 
 Again $N= 3^d\,{n \choose d}$.
 For $2^n \times 2^n$ matrix $X=\frac1{\sqrt{N}} \sum_{j=1}^N g_j \Sigma_j$ we recall that $g_j$ are i.i.d. standard gaussians
 and $\Sigma_j$ are tensor products of elementary Paulis among which only exactly $d$ are not identities.
 
 In \cite{BBoVH} the following important characteristics of  random matrix $X$ are  introduced and used.  One  of them called $\Delta$ we already saw above (normalization is different from one in \cite{BBoVH}):
 $$
 \sigma_*(X)^2 := \max_{\|u\|=\|v||=1} \frac1{N}\sum_{j=1}^N |\langle \Sigma_j u, v\rangle|^2\,.
 $$
 In fact,  for our system of Paulis
 $$
 \cS_d:=\{ \Sigma_j\}_{j=1}^N
 $$
 we used
 $$
 \Delta(\cS_d) = \max_{\|u\|=1} \frac1{N} \sum_{j=1}^N  |\langle \Sigma_j u, u\rangle|^2\asymp \max_{\rho\ge 0, \tr\rho=1} \frac{1}{N}\sum_{j=1}^n |\tr[\Sigma_j \rho]|^2\,.
 $$
 The quantities $\sigma_*(X)^2$ and $\Delta(S_d)$ are comparable with absolute constants just by polarization.
 
 \medskip
 
 The following characteristic $v(X)$ plays an important part in \cite{BBoVH}:
 $$
 v(X)^2 = \max_{\tr [|M|^2] \le 1} \frac1N \sum_{j=1}^N |\tr [\Sigma_j M]|^2\,.
 $$
 Lemma 4.7 of \cite{BBoVH} proves that
 
 \begin{equation}
 \label{HS}
 v(X) = \max_j \frac1{\sqrt{N}}\|\Sigma_j\|_{HS}= \frac{2^{n/2}}{N^{1/2}},
 \end{equation}
 as $\tr[\Sigma_k\Sigma_j]=0, k\neq j$.
 We will give the proof for the sake of completeness below.
 
 \medskip
 
 In \cite{BBoVH} it is proved that for $\bE \|\bH(d, n)\|_{op}$ the second term of the asymptotics
 is bounded by
 $$
 v(X)^{1/2} \sigma(X)^{1/2} \log ^{3/4} D,
 $$ 
 where $D=2^n$ the dimension of the Hilbert space and 
 $$
 \sigma(X) = \frac1N \|\sum_{j=1}^N\Sigma_j^2\|_{op} = \|\text{Id}_{2^n\times 2^n}\|_{op} =1\,.
 $$
 
 Gathering all above we get that the second term of the asymptotic is
 $$
B(d, n)= \frac{ 2^{n/2}\, n^{3/4}}{ 3^{d/2} \, {n \choose d}^{1/2}}\,.
 $$
 If $d= pn$ this is approximately
 $$
 n^{3/4}2^{\frac12 n-\frac12 \frac{\log 3}{\log 2} d - n (p \log \frac1p + (1-p) \log\frac1{1-p})}
 $$
 and is exponentially small if  $d \ge c_0 n$ for a certain $c_0\in (0,1)$. This $c_0$ is obviously $< \frac{\log 2}{\log 3}$.
 
 \medskip
 
 So for such large $d$ we have that
 \begin{equation}
 \label{large-d}
 \bE \|\bH(d, n)\|_{op} \asymp C\sqrt{n}, \quad  \bE \|H(d, n)\|_{op} \asymp C \,3^{d/2}\,\sqrt{n}
 \end{equation}
 
 However, for, say, $d^2 \le b_0 n$ for a small absolute $b_0$ we have two estimates that do not match:
 So for such 
 small $d$ we have that
 \begin{align}
 \label{small-d}
&\frac{c\sqrt{n}}{d}  \le \bE \|\bH(d, n)\|_{op} \le C\sqrt{n},  \,\,\,\text{which is}
\\
& \frac{c\, 3^{d/2}\sqrt{n}}{d} \le \bE \|H(d, n)\|_{op} \le c \,3^{d/2}\,\sqrt{n}\,.
 \end{align}
 
 \subsection{Equality \eqref{HS} }
Equality \eqref{HS} is proved in Lemma 4.7 of \cite{BBoVH}. For the sake of completeness we give the reasoning now (virtually the same).
 It is easy to see that $v(X)^2$ is also the norm of the $D^2\times D^2$ matrix with $(ij, kl)$ element equal to $\frac1N \sum_{m=1}^N (\Sigma_m)_{ij} (\Sigma_m)_{kl}$. Denoting pairs of indices by letters $\mu, \nu$
 and thinking about matrices $\Sigma_m$ as $\ell_2^{D^2}$- vectors $V_1,\dots, V_N$ this is nothing else  as $\frac1N \sum_{m=1}^N (V_m)_\mu (V_m)_\nu$. 
 
 The norm of rank-$1$ matrix $(\cdot, V_m) V_m$ (it has exactly these $(\mu, \nu)$ matrix elements: $(V_m)_\mu (V_m)_\nu$) is $\|V_m\|_2^2$, which is
 $\|\Sigma_m\|_{HS}^2$. But we have the sum of those rank-$1$ matrices. However, $\tr [\Sigma_m\Sigma_{m'}] = 0$ for $m\neq m'$ means 
 that vectors $V_m$ are orthogonal in $\ell_2^{D^2}$. Therefore,
 the norm of the sum is the maximal of norms: $\frac1N \max_{1\le m\le N} \|V_m\|_2^2 =  \frac1N \max_{1\le m\le N} \|\Sigma_m\|_{HS}^2$.
 Therefore \eqref{HS} holds, that is
 $$
 v(X) =\frac1{\sqrt{N}} \max_{1\le i \le N} \|\Sigma_i\|_{HS}\,.
 $$

 \bigskip
 
\subsection{Estimating $\Delta(\cS_d)$}
 Paper \cite{ACKK} has very beautiful calculation of $\Delta(\cS_d)$. We repeat it here.
 
 \begin{lemma} 
 \label{Delta}
 If $3^d \le 2n+1$ then $\Delta(\cS_d) \le 3^{-d}$.
 \end{lemma}

 \begin{proof}
 First one should notice that
 for a Hermitian matrix $A$ and a unit vector $\psi$ it is always true that
 \begin{equation}
 \label{Apsi}
 \big( \langle A\psi, \psi\rangle \big)^2 \le \langle A^2 \psi, \psi\rangle\,.
 \end{equation}
 Inequality \eqref{Apsi} follows easily from H\"older inequality.
 Then \eqref{Apsi} can be easily upgraded to the following inequality, where $\rho$ is any state:
 \begin{equation}
 \label{Arho}
 \big(\tr [A\rho]\big)^2 \le \tr [A^2 \rho]\,.
 \end{equation}
 
 We can apply \eqref{Arho} to 
 $$
 A=\sum_{j=1}^N \tr[\Sigma_j \rho] \, \Sigma_j.
 $$
 Then we get the inequality like this:
 $$
 \big[\sum_{j=1}^N (\tr [\Sigma_j \, \rho])^2\big]^2 \le \sum_{k, j} \tr [\Sigma_k \, \rho] \tr [\Sigma_j \, \rho]\, \tr [\Sigma_k \Sigma_j \rho] = 
 $$
 $$
\frac12 \sum_{k, j} \tr [\Sigma_k \, \rho] \tr [\Sigma_j \, \rho]\, \tr [(\Sigma_k \Sigma_j +\Sigma_j\Sigma_k)\rho] 
 $$
 Anti-commuting $\Sigma'$s give $0$, but the same $\Sigma'$s give us $1$ as $\tr \rho=1$ and $\Sigma_j^2=\text{Id}$. 
 So, we see the inequality below that was first proved in \cite{GHG}, see also \cite{AGK}, \cite{ACKK}: 
 \begin{equation}
 \label{B}
\max_{\rho\ge 0, \tr \rho=1} \sum_{j=1}^n (\tr [\Sigma_j\rho])^2 \le \max \|B\|_{op},
 \end{equation}
 where $B$ is non-negatively defined, $B_{kk}=1$, $B_{kj}=0$ if $\Sigma_k, \Sigma_j$ anti-commute. Looks like an adjacent matrix of a graph.  
 Such maximum is called Lov\'asz number of anti-commuting graph $G(\cS_d)$ of $\cS_d$.
  Graph $G(\cS_d)$ has $\cS_d$ as vertices and there is no edge if $\Sigma_k, \Sigma_j$ anti-commute.
 Lov\'asz number is denoted by $\theta(G(\cS_d))$.
 
 For any graph $G$, $\bar G$ denoted complemented graph--same vertices and exactly those edges which were not present in $G$.
 
 Graph $\bar G(\cS_d)$ is obviously transitive (there is an automorphism sending any vertex to any vertex connected by an edge). 
 For transitive graphs there is an equality
 $$
 \theta(G) \theta(\bar G) = \sharp{\text{vertices}}
 $$
 So
 \begin{equation}
 \label{th}
 \theta(G(\cS_d)) \theta(\bar G(\cS_d))  =N
 \end{equation}
 And we already saw in \eqref{B} that
 $$
  \Delta(S_d) \le \frac 1N \theta(G(\cS_d)) =(\theta(\bar G(\cS_d)))^{-1}\,.
 $$
 
 It is time to use a combinatorial result of Knuth \cite{K}, that says that $\theta (g)$ can be estimated from below by the size of maximal clique of $g$. 

We can illustrate this by two extreme examples. One is when $g$ is a full graph on $n$ vertices. 
Then $\theta(g)$ is the norm of matrix whose entries are all $1$. In this case $\theta(g)= n$, this is the size of maximal cliques obviously.
The second example is $g$, where there are no edges, but there are $n$ vertices. Again $\theta(g)=1$, which again is the size of the maximal clique.

 Graph
 $\bar G(\cS_d)$ has edges between any two $\Sigma_k, \Sigma_j$ that anti-commute. It has been known \cite{JKMN} that the maximum clique has at least $3^d$ vertices if 
 $$
 2n+1 \ge 3^d\,.
 $$
 Let us illustrate this for $d=2,  n=4$.
 
 $$
 (1001), (1002), (1003), (2010), (2020), (2030), (3100), (3200), (3300)
 $$
 are nine tensor products $\sigma_1\otimes\sigma_0\otimes \sigma_0\otimes\sigma_1,\dots , \sigma_3\otimes\sigma_3\otimes \sigma_0\otimes\sigma_0$ that anti-commute if we have $n=4$ qubits Hilbert space.
 \end{proof}
 
\medskip

 \begin{lemma} 
 \label{Delta23}
 Always $\Delta(\cS_d) \le \Big(\frac23\Big)^{d}$.
 \end{lemma}
 
\begin{proof}
This lemma is also proved in \cite{ACKK}. Our proof is only slightly different. We fix 
a unit vector $\psi$ in $\bC^{2^n}$:
$$
\psi= \sum_{x\in \{0,1\}^n} \alpha_{x} \psi_{x_1}\otimes\dots\otimes \psi_{x_n}=: \sum \alpha_x \psi_x,
$$
where $\psi_{x_i} $ is $\ket 1$ if $x_i=1$ and $\ket 0$ if $x_i=0$. 

Our quantity
$$
\Delta(\cS_d)=\bE _{\Sigma\in \cS_d} \langle \Sigma\psi, \psi\rangle^2=
$$
$$
\bE_{S\subset [n], |S|=d} \bE_{\cS_d^S}\langle \Sigma\psi, \psi\rangle^2 = \bE_{S\subset [n], |S|=d}\frac1{3^d} \sum_{\Sigma \in \cS_d^S}\langle \Sigma\psi, \psi\rangle^2  \,.
$$
We need to estimate $\frac1{3^d} \sum_{\Sigma \in \cS_d^S}\langle \Sigma\psi, \psi\rangle^2 $ for a fixed $S$. Let us allow to have also $\sigma_0$, and call the new collection of Paulis aimed at $S\subset [n]$ by
$\cP_d^S$. Then we continue thinking that 
$$
S=[d],
$$
(which is without loss of generality as all $S$ are the same) to get
\begin{align}
\label{Jd}
 &\frac1{3^d} \sum_{\Sigma \in \cS_d^S}\langle \Sigma\psi, \psi\rangle^2  \le  \frac1{3^d} \sum_{\Sigma \in \cP_d^S}\langle \Sigma\psi, \psi\rangle^2  = \notag
 \\
 & \frac1{3^d}\bra\psi\bra\psi \sum_{\Sigma \in \cP_d^S}\Sigma\otimes\Sigma\ket\psi\ket\psi=: \frac{J}{3^d}.
 \end{align}
There are two independent $\psi\otimes\psi$ here. We think that the first one and the second one are correspondingly
$$
\psi\otimes \psi = \sum_{x, x'} \alpha_{x}\alpha_{x'} \psi_x\otimes \psi_{x'},\,\,\psi\otimes \psi = \sum_{y, y'} \alpha_{y}\alpha_{y'} \psi_{y}\otimes\psi_{y'}
$$
Before plugging these expressions into \eqref{Jd} let us recall what is the operator called $\text{SWAP}$. 
This is the operator in $\bC^2\otimes \bC^2$ that sends $\ket u\ket v$ into $\ket v\ket u$. Its matrix is 
$$
\text{SWAP}= \begin{bmatrix} 
1\,0\,0\,0
\\
0\,0\,1\,0
\\
0\,1\,0\,0
\\
0\,0\,0\,1
\end{bmatrix}
$$
It is well known that
$$
2\text{SWAP} = \sigma_0\otimes \sigma_0 +  \sigma_1\otimes \sigma_1+  \sigma_2\otimes \sigma_2 + \sigma_3\otimes \sigma_3\,,
$$
which we will use soon below.
 \begin{align*}
& J = \sum_x\sum_y \sum_{x'} \sum_{y'} \alpha_x\alpha_{y} \alpha_{x'}\alpha_{y'}\prod_{i=1}^d \sum_{\kappa\in \{0, 1, 2, 3\}} \Big(\bra{\psi_{x_i}} \sigma_\kappa \ket{\psi_{y_i}}\bra{\psi_{x'_i}} \sigma_\kappa \ket{\psi_{y'_i}}\Big)\cdot
\\
&\prod_{j=d+1}^n\Big(\bra{\psi_{x_j}} \sigma_0 \ket{\psi_{y_j}}\bra{\psi_{x'_j} }\sigma_0 \ket{\psi_{y'_j}}\Big) =
\\
&  \sum_x\sum_y \sum_{x'} \sum_{y'} \alpha_x\alpha_{y} \alpha_{x'}\alpha_{y'}\prod_{i=1}^d \bra{\psi_{x_i}}\bra{\psi_{x'_i}}\sum_{\kappa\in \{0, 1, 2, 3\}}  \sigma_\kappa\otimes \sigma_\kappa \ket{\psi_{y_i}}\ket{\psi_{y'_i}} \cdot
\\
&\prod_{j=d+1}^n \bra{\psi_{x_j}}\bra{\psi_{x'_j}}  \sigma_0\otimes \sigma_0 \ket{\psi_{y_j}}\ket{\psi_{y'_j}} =
\\
&  \sum_x\sum_y \sum_{x'} \sum_{y'} \alpha_x\alpha_{y} \alpha_{x'}\alpha_{y'}\prod_{i=1}^d \bra{\psi_{x_i}}\bra{\psi_{x'_i}}2\text{SWAP}_{i, i+n} \ket{\psi_{y_i}}\ket{\psi_{y'_i}} \cdot
\\
&\prod_{j=d+1}^n \bra{\psi_{x_j}}\bra{\psi_{x'_j}}  \sigma_0\otimes \sigma_0 \ket{\psi_{y_j}}\ket{\psi_{y'_j}} =
\\
&2^d\, \sum_x\sum_y \sum_{x'} \sum_{y'} \alpha_x\alpha_{y} \alpha_{x'}\alpha_{y'}  \prod_{i=1}^d \bra{\psi_{x_i}}\bra{\psi_{x'_i}}\ket{\psi_{y'_i}}\ket{\psi_{y_i}} \cdot
\\
&\prod_{j=d+1}^n \bra{\psi_{x_j}}\bra{\psi_{x'_j}}  \sigma_0\otimes \sigma_0 \ket{\psi_{y_j}}\ket{\psi_{y'_j}}\,.
 \end{align*}
 
 Huge number of products vanish. Let us list those that survive. These are those, where $x_i= y'_i$ and $y_i = x'_i$, for all $i=1,\dots, d$ and 
 simultaneously $x_j= y_j$ and $y'_j= x'_j$, $i= d+1,\dots, n$.
 
 \medskip
 
 Rewrite all quadruple products $\alpha_x\alpha_{y} \alpha_{x'}\alpha_{y'}$ as follows (below $S=[d], \bar S=[n]\setminus [d]$):
 $$
 \alpha_{x|S\, x|\bar S}\, \alpha_{y|S\, y|\bar S}\, \alpha_{x'|S\, x'|\bar S}\, \alpha_{y'|S\, y'|\bar S}\,.
 $$
 
 So only those quadruple products survive 
 $$
 \alpha_{x|S\,x|\bar S}\, \alpha_{x'|S\, x|\bar S}\, \alpha_{x'|S\, x'|\bar S}\, \alpha_{x|S\,x'|\bar S}
 $$
 To simplify we use nee notations and we write the sum of surviving products as
 $$
 \sum_{I, J, K, L} \alpha_{IJ}\, \alpha_{KJ} \, \alpha_{KL}\, \alpha_{IL}\,.
 $$
 Let us denote $A= \{\alpha_{IJ}$ the $d\times (n-d)$ matrix of entries $\alpha-{x|S\, x|\bar S}$ (again $S=[d], \bar S =[n]\setminus [d]$).
 Then 
 $$
 \|A\|_{HS} =\|\{\alpha_x\}\|_2 =1,
 $$ 
 and
 $$
 \sum_{I, J, K, L} \alpha_{IJ}\, \alpha_{KJ} \, \alpha_{KL}\, \alpha_{IL}=\tr[ A^*A\cdot A^*A] \le \|A\|_{HS}^4 = 1\,.
 $$
 
 Henceforth,
 $$
 J\le 2^d,
 $$
 and lemma is proved because of \eqref{Jd}.

\end{proof}

\section{One impossibility and questions}
\label{e}
Notice that if we have a good estimate of this quantity from below then we will have a good estimate for 
(a smaller) quantity
$$
 e^{\beta\sqrt{n}\bE \|H(d, n)\|_{op}}
 $$
 from below. 
\begin{equation}
\label{March}
 e^{\beta\sqrt{n}\bE \|H(d, n)\|_{op}} \ge \bE e^{\beta\sqrt{n} \|H(d, n)\|_{op}} \cdot\text{not too small},
 \end{equation}
 where in the case
 \begin{equation}
 \label{d-n}
 3^d \le 2n+1
 \end{equation}
 we can consult \eqref{eE} and see that 
 $\text{not too small}\ge e^{-4n\beta^2}$.
 That is
 \begin{equation}
\label{March0}
  3^d \le 2n+1 \, \Rightarrow\,e^{\beta\sqrt{n}\bE \|H(d, n)\|_{op}} \ge \bE e^{\beta\sqrt{n} \|H(d, n)\|_{op}} \cdot e^{-4 n\beta^2}\,.
 \end{equation}
 See the detailed explanation  in Section \ref{explan}.
 
 \bigskip

 To estimate $\bE e^{\beta\sqrt{n} \|H(d, n)\|_{op}}$ from below we first use
 \begin{equation}
 \label{tr}
 \bE e^{\beta\sqrt{n} \|H(d, n)\|_{op}} \ge 2^{-n} \bE\tr e^{\beta\sqrt{n} H(d, n)} \,.
 \end{equation}
 Expanding $e^x$ and noticing that odd powers are eliminated by $\bE$ we now need the estimate from below for
 \begin{align*}
 &\sum_{m=0}^\infty \frac1{(2m)!}\beta^{2m} n^m\bE \tr [H(d, n)^{2m}]= \sum_{m=0}^{C_1n}\dots + 
 \\
 &\sum_{m>C_1 n} \frac1{(2m)!}\beta^{2m} n^m\bE \tr [H(d, n)^{2m}]=: \sigma_1+\sigma_2\,.
\end{align*}

\medskip

Lemma 39 of \cite{AGK} claims that for every $d, \eps, \beta$ there exists $C_1(d, \eps, \beta) $ such that $\sigma_2= e^{-\eps n}$. 
This should allow us to consider only $\sigma_1$.

\medskip

Now let us take a look at the second inequality of \eqref{asymptH}. We proved this estimate  only for $m\le c\frac{n}{d^2}$:
\begin{equation}
\label{trtr}
\bE \tr [H(d, n)^{2m}] \ge 2^n \, c^m\frac{(2m)!}{m!} 3^{dm}
\end{equation}

Suppose we would be able to prove it for $m\le C_1n$ with $C_1$ as large as above, then
 using \eqref{tr} and \eqref{March0}
  we would have 
 \begin{equation}
 \label{March1}
\bE \big( \frac1{\sqrt{n}} \|H(d, n)\|_{op}\big) \ge c_0 3^{d/2} \beta - \frac{4n\beta^2 }{\beta n} =  c_0 3^{d/2} \beta- 4\beta \,,
\end{equation}
which is nonsense as $\beta$ can be chosen arbitrarily.

What else do we know except
$$
\sqrt{n}\frac{3^{d/2}}{d} \le \bE\|H(d, n)\|_{op} \le \sqrt{n} \,3^{d/2}\,?
$$
For $d\asymp 1$ this is a sort of equivalence.

For ``commutative case" $ 1<< d\le c\sqrt{n}$ the upper estimate should be improved to the lower one, but we did not prove that.

If $d\ge c\sqrt{n}$ the lower estimate becomes weaker than noncommutative Khintchine inequality (NCK) estimate $ c \le \bE\|H(d, n)\|_{op}$. 

What happens for $c\sqrt{n} \le d\le c_1 n$? 

Looks like we do not have any information, but it is reasonable to think that the left hand side of NCK estimate becomes prevailing for $d\asymp \sqrt{n}$.

For large $d$, like $d\asymp n$, the highly ``noncommutative case" we should have 
$$
 \bE\|H(d, n)\|_{op}  \asymp 3^{d/2}\sqrt{n}\,.
 $$

Another question: for a fixed $d$ is  there a limit:
$$
\lim_{n\to \infty} \frac1{\sqrt{n}} \bE\|H(d, n)\|_{op}\,?
$$

\newcommand{\etalchar}[1]{$^{#1}$}

\end{document}